\documentclass[12pt, draftclsnofoot, onecolumn]{IEEEtran}

\usepackage{color}
\usepackage{amsthm}
\usepackage{cite}
\usepackage{amssymb}
\usepackage{amsmath}
\usepackage[lined, boxed, commentsnumbered, ruled, linesnumbered]{algorithm2e}
\usepackage{mathrsfs}
\usepackage{bm}
\usepackage{subcaption}
\usepackage{pgfplots}
\usepackage{tikz}
\usetikzlibrary{arrows}
\usepackage{graphicx,booktabs,multirow}
\usepackage{hyperref}
\usepackage{setspace}
\usepackage{indentfirst}

\newtheorem{lemma}{Lemma}
\newtheorem{theorem}{Theorem}

\newtheorem{assumption}{Assumption}
\theoremstyle{definition} 
\newtheorem{remark}{Remark}

\usepackage[labelformat=simple]{subcaption}
\allowdisplaybreaks
\newcommand{\trace}{{\mathrm{Tr}}}

\newcommand{\R}{\mathbb{R}}  % The real numbers.
\newcommand{\C}{\mathbb{C}} % The complex numbers.

\usepackage{soul,xcolor}
\setstcolor{red}
\newcommand{\change}[1]{\textcolor{blue}{#1}}

\usepackage[explicit]{titlesec}
 
\titlespacing*{\section}{0pt}{1.2ex plus .0ex minus .0ex}{.3ex plus .0ex}
\titlespacing*{\subsection}{0pt}{1.2ex plus .0ex minus .0ex}{.3ex plus .0ex}

\begin{document}
\title{Vertical Federated Learning over Cloud-RAN: Convergence Analysis and System Optimization}
\author{\IEEEauthorblockN{
        Yuanming Shi,~\IEEEmembership{Senior Member,~IEEE,}
        Shuhao Xia,~\IEEEmembership{Student Member,~IEEE,}
        Yong Zhou,~\IEEEmembership{Senior Member,~IEEE,}
        Yijie Mao,~\IEEEmembership{Member,~IEEE,}
        Chunxiao Jiang,~\IEEEmembership{Senior Member,~IEEE,}
        and Meixia Tao,~\IEEEmembership{Fellow,~IEEE}
        }
%     \thanks{
% 		S. Xia is with the School of Information Science and Technology, ShanghaiTech University, Shanghai 201210, China, also with the Shanghai Institute of Microsystem and Information Technology, Chinese Academy of Sciences, Shanghai 200050, China, and also with the University of Chinese Academy of Sciences, Beijing 100049, China (e-mail: xiashh@shanghaitech.edu.cn).}
	\thanks{
		Y. Shi, S. Xia, Y. Zhou, and Y. Mao are with the School of Information Science and Technology, ShanghaiTech University, Shanghai 201210, China (e-mail: \{shiym, xiashh, zhouyong, maoyj\}@shanghaitech.edu.cn). \emph{(Corresponding author: Yong Zhou.)}}
	\thanks{
		C. Jiang is with the Tsinghua Space Center and the Beijing National Research Center for Information Science and Technology, Tsinghua University, Beijing 100084, China (e-mail: jchx@tsinghua.edu.cn).}
%	\thanks{
%		Y. jiang is with Automatic Control Laboratory, EPFL, Lausanne, Switzerland (e-mail: yuning.jiang@epfl.ch).}
	\thanks{
		M. Tao is with Department of Electronic Engineering, Shanghai Jiao Tong University, Shanghai 201210, China (e-mail:mxtao@sjtu.edu.cn).}
 }
\maketitle

\setlength{\textfloatsep}{2pt} 
\setlength\abovedisplayskip{2pt}
\setlength\belowdisplayskip{2pt}

\vspace{-1cm}
\begin{abstract}
Vertical federated learning (FL) is a collaborative machine learning framework that enables devices to learn a global model from the feature-partition datasets without sharing local raw data. 
However, as the number of the local intermediate outputs is proportional to the training samples, it is critical to develop communication-efficient techniques for wireless vertical FL to support high-dimensional model aggregation with full device participation. 
In this paper, we propose a novel cloud radio access network (Cloud-RAN) based vertical FL system to enable fast and accurate model aggregation by leveraging over-the-air computation (AirComp) and alleviating communication straggler issue with cooperative model aggregation among geographically distributed edge servers. 
However, the model aggregation error caused by AirComp and quantization errors caused by the limited fronthaul capacity degrade the learning performance for vertical FL.  
To address these issues, we characterize the convergence behavior of the vertical FL algorithm considering both uplink and downlink transmissions.
To improve the learning performance, we establish a system optimization framework by joint transceiver and fronthaul quantization design, for which successive convex approximation and alternate convex search based system optimization algorithms are developed. 
We conduct extensive simulations to demonstrate the effectiveness of the proposed system architecture and optimization framework for vertical FL.
\end{abstract}

\vspace{-0.6cm}
\begin{IEEEkeywords}
% Vertical federated learning, cloud radio access network (Cloud-RAN), over-the-air computation (AirComp), convergence analysis, resource allocation, limited fronthaul capacity.
\vspace{-0.4cm}
Vertical federated learning, cloud radio access network, over-the-air computation.
\end{IEEEkeywords}

\section{Introduction}
Federated learning (FL), as an emerging distributed learning paradigm, has recently attracted lots of attention by enabling devices to collaboratively learn a global machine learning (ML) model without local data exchanging to protect data privacy \cite{DBLP:journals/jsac/LetaiefSLL22}.
Based on the partition of datasets, FL is typically categorized into horizontal FL and vertical FL \cite{DBLP:journals/tist/YangLCT19}.
Specifically, horizontal FL is adopted in the sample-partition scenario where different devices share the same feature space but have different samples \cite{DBLP:journals/comsur/LimLHJLYNM20}, as shown in Fig. \ref{fig: hfl}.
Vertical FL, on the other hand, enables devices to collaboratively learn a global model from the feature-partition datasets that share the same sample space but different feature spaces \cite{DBLP:conf/nips/yang19}, as shown in Fig. \ref{fig: vfl}.
Vertical FL has wide applications across e-commerce, smart healthcare, and Internet of Things (IoT).
For example, in IoT networks, data collected by different types of sensors (e.g., video cameras, GPS, and inertial measurement units) need to be smartly fused for planning and decision-making \cite{DBLP:journals/iotj/LiangYL19, DBLP:journals/icl/LiuZJLXC22}.
As each participant owns partial model in vertical FL, the messages exchanged between participants are based on the intermediate outputs of the sub-model and its local data, rather than local updates as in horizontal FL.
The number of intermediate outputs is proportional to the number of training samples.
This leads to high-dimensional data transmission with a large volume of training samples \cite{DBLP:journals/corr/abs-2202-04309}.
In addition, since vertical FL needs to aggregate all local intermediate outputs to obtain the final prediction, it requires all devices to participate \change{in} the training process in vertical FL. 
Therefore, the high-dimensional data transmission with full device participation in vertical FL brings a unique challenge for communication-efficient system design.
Existing works on device selections in horizontal FL cannot address this issue in wireless vertical FL.

To design communication-efficient wireless vertical FL, it is critical to develop novel wireless techniques to transmit intermediate outputs over shared wireless medium with limited radio resources \cite{DBLP:journals/jsac/ChenGHSBFP21a}.
Conventional orthogonal multiple access schemes aim to decode individual information, which however ignores the communication task for data transmission.
Instead, over-the-air computation (AirComp) serves as a promising approach to enable efficient aggregation for FL \cite{DBLP:journals/cm/LetaiefCSZZ19, DBLP:journals/twc/YangJSD20, 9843892, yuanming2023}.
By exploiting waveform superposition of multiple access channel, AirComp can directly aggregate the local updates via concurrent uncoded transmission \cite{DBLP:journals/wc/ZhuXHC21}.
As a result, AirComp can reduce the required radio resources and the access latency compared with orthogonal multiple access schemes \cite{DBLP:journals/comsur/ShiYJZL20}.
To further improve the communication and learning efficiency of wireless FL with AirComp, the intrinsic properties of the local updates are exploited, e.g., gradient sparsity \cite{DBLP:journals/tsp/AmiriG20}, temporal correlation \cite{DBLP:journals/jsac/FanYZ21}, spatial correlation \cite{DBLP:journals/corr/abs-2112-13603}, and gradient statistics \cite{DBLP:journals/twc/ZhangT21a}.
However, due to the heterogeneity of channel conditions across different devices, AirComp based FL inevitably suffers from the communication stragglers, i.e., the devices with the weak channel conditions \cite{DBLP:journals/iotj/GuoLL21, DBLP:conf/icc/XiaZY0S021}.
Specifically, the communication straggler devices enforcing the magnitude alignment for AirComp may result in high model aggregation error, which degrades the learning performance. 
Although it can be alleviated by excluding the stragglers from the model aggregation process \cite{DBLP:journals/twc/ZhuWH20, DBLP:journals/twc/LiuYZ21}, partial device participation reduces the size of the FL training datasets, which deteriorates the performance for FL models.
Besides, these works mainly focus on horizontal FL with partial device selection and are inapplicable in the vertical FL scenario.
The reason is that vertical FL requires all the devices to participate the training process since each device owns one sub-model of the global model. 
It is thus critical to design a novel AirComp-based transmission scheme for wireless vertical FL with full device participation. 

\begin{figure}[t!]
\centering
\begin{subfigure}[h]{0.49\textwidth}
        \centering
        \includegraphics[scale=0.7]{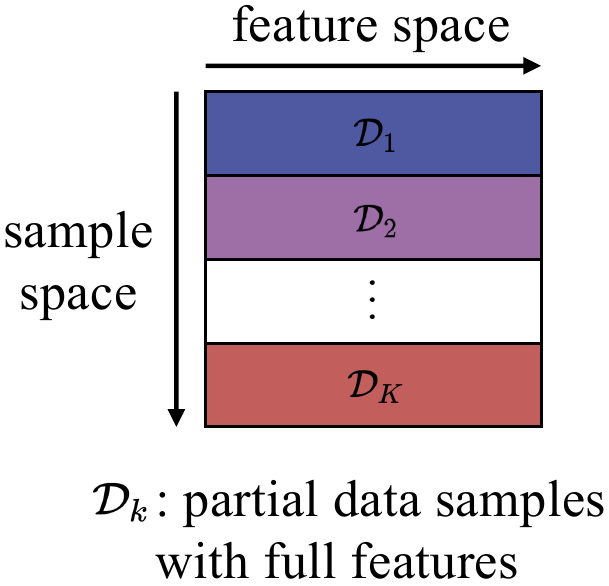}
        \caption{Horizontal FL}
        \label{fig: hfl}
\end{subfigure}
\begin{subfigure}[h]{0.49\textwidth}
        \centering
        \includegraphics[scale=0.7]{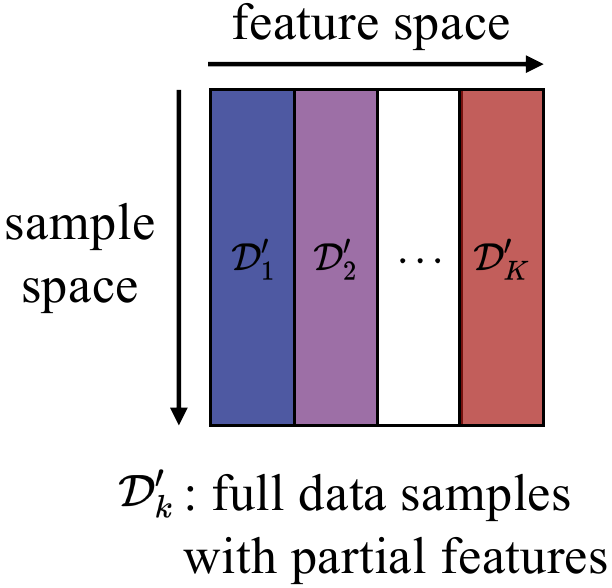}
        \caption{Vertical FL}
        \label{fig: vfl}
\end{subfigure}
\caption{Sample and feature spaces of horizontal and vertical FL.}
\label{fig: categorization}
% \vspace{-0.8cm}
\end{figure}

To enable all devices including the straggler devices to participate in the training process, one approach is to improve the wireless channel quality through advanced communication technologies such as unmanned aerial vehicles \cite{DBLP:journals/jsac/YangZXLSX21} and reconfigurable intelligent surface \cite{DBLP:journals/twc/WangQZSFCL22, DBLP:journals/corr/abs-2109-02353, DBLP:journals/network/YangSZYFC20}.
However, these wireless techniques may become inapplicable in the IoT scenarios with massive geographically distributed devices.
An alternative approach is to effectively extend the wireless transmission coverage by deploying multiple edge servers, thereby cooperatively assisting data exchanges between the devices and the edge servers  \cite{DBLP:journals/twc/VuNTNDM20, DBLP:journals/corr/abs-2107-09518, DBLP:journals/jsac/Wang0SZ22}.
This can shorten the communication distances between the devices and the edge servers, which enhances the transmission reliability and reduces the latency for the straggler devices.
A growing body of recent work has demonstrated the effectiveness of deploying multiple edge servers for FL, including semi-decentralized FL \cite{DBLP:journals/jsac/LinHABM21} and hierarchical FL \cite{DBLP:journals/jsac/LimNXNMK21, DBLP:journals/jsac/ZhangYWPZZS21}.
To mitigate inter-cluster interference, the authors in \cite{DBLP:journals/jsac/LinHABM21, DBLP:journals/jsac/LimNXNMK21} proposed to allocate orthogonal radio resource blocks for each device for model aggregation in FL.
However, this orthogonal transmission scheme is inefficient due to the limited spectrum resources.
In addition, many studies have considered distributed edge computing systems based on Cloud-RAN architecture with various communication metrics, such as energy consumption \cite{DBLP:journals/tvt/ChenCGLHH21} and latency \cite{DBLP:journals/tsipn/ParkJNSS21}, for resource management.

In this paper, we propose to leverage cloud radio access network (Cloud-RAN) to enable cooperative model aggregation across densely deployed multiple edge servers over large geographic areas.
This is achieved by moving baseband processing from remote radio heads (RRHs) (i.e., edge servers) to a baseband unit (BBU) pool (i.e., central server), where digital fronthaul links connect the central server and the edge servers \cite{DBLP:journals/wc/ShiZL0C15}.
Therefore, the learning performance and communication efficiency can be significantly improved through cooperative model aggregation and centralized signal processing for large-scale wireless vertical FL system.
Besides, since the baseband signal processing and model aggregation are moved to the central server in the BBU pool, edge servers only need to support basic signal transmission functionality, which further reduces their energy consumption and deployment cost \cite{DBLP:journals/wc/ShiZL0C15, DBLP:journals/twc/TaoCZ016}.

In addition, we propose a fast and reliable transmission scheme to support vertical FL with geographically distributed devices over wireless networks. 
% In particular, the global model aggregation in vertical FL consists of the transmission of local intermediate outputs from each device, followed by the computation of model aggregation at the central server via AirComp.
Different from the existing works on vertical FL over wireless networks \cite{DBLP:journals/iotj/SuL21}, we shall implement AirComp-based vertical FL over Cloud-RAN to support heterogeneous devices.
% We highlight that the system model considered in \cite{DBLP:journals/iotj/SuL21} is a special case of the model considered in this work when the number of cell is reduced to 1 and the fronthaul capacity goes to infinity.
However, the fronthaul links between the edge servers and the central server have limited capacity, which causes additional quantization error. 
Furthermore, the quantization error along with AirComp model aggregation error severely degrade the learning performance of vertical FL. 
The main purpose of this paper is to reveal the impact of limited fronthaul capacity and AirComp aggregation error on the learning performance and then optimize the Cloud-RAN enabled vertical FL system by jointly considering the learning performance and communication efficiency. 
The main contributions are summarized as follows:
\begin{itemize}
    \item We propose a Cloud-RAN based network architecture for communication-efficient vertical FL to alleviate the communication straggler issue by leveraging the distributed edge servers with shortened communication distance between devices and edge servers. 
    This is achieved by transmitting the intermediate outputs at each device to multiple edge servers (i.e., RRHs) via AirComp.
    The edge servers then compress the received aggregation signals and forward the quantized signals to the central server via digital fronthaul links for centralized global model aggregation in vertical FL.
	\item We characterize the convergence behaviors of the gradient-based vertical FL algorithm under AirComp model aggregation error and limited fronthaul capacity in both uplink and downlink transmission processes.
	We further propose a system optimization framework with joint transceiver and quantization design to minimize the learning optimality gap.
	The resulting system optimization problem is solved based on the successive convex approximation and alternate convex search approaches with convergence guarantees.
	\item We conduct extensive simulations to demonstrate the effectiveness of the proposed system optimization framework and system architecture. 
	We demonstrate that the proposed framework can effectively combat the effects of communication limitations on the learning performance of vertical FL in terms of the optimality gap. 
	The superiority of the distributed antenna system in Cloud-RAN for vertical FL is also demonstrated. 
\end{itemize}

% Overall, we demonstrate that the proposed Cloud-RAN system based on AirComp is a promising and powerful framework that alleviates the communication straggler issue in large-scale vertical FL, thereby supporting scalable and fast aggregation.

The remainder of this paper is constructed as follows.
Section II introduces the learning framework of vertical FL and the communication model.
Section III provides the convergence analysis of the vertical FL algorithm, and Section IV formulates the system optimization problem for resource allocation.
Extensive numerical results of the proposed scheme are presented in Section V, followed by the conclusion in Section VI.

Throughout this paper, we denote the complex number set and the real number set by $\mathbb{C}$ and $\mathbb{R}$.
We denote scalars, vectors, and matrices by regular letters, bold small letters, and bold capital letters, respectively.
The operations of transpose and conjugate transpose are respectively denoted by the superscripts $(\cdot)^\mathsf{T}$ and  $(\cdot)^\mathsf{H}$.
$\mathbb{E}[\cdot]$ denotes the expectation operator. $\mathcal{N}(\bm{x};\bm{\mu}, \bm{\Sigma})$ and $\mathcal{CN}(\bm{x};\bm{\mu}, \bm{\Sigma})$ denote that the random variable $\bm{x}$ follows Gaussian distribution and complex Gaussian distribution with mean $\bm{\mu}$ and covariance $\bm{\Sigma}$, respectively.
We use $\bm{I}$ and $\operatorname{diag}\{\bm{x}\}$ to respectively denote the identity matrix and the diagonal matrix with diagonal entries specified by $\bm{x}$.
We denote real and imaginary components of $x$ by $\Re(x)$ and $\Im(x)$, respectively.

% \section{Vertical Federated Learning over Cloud Radio Access Network}
\section{System Model}
\subsection{Vertical Federated Learning}\label{sec: vertical fl}
\begin{figure}[h]
\centering
	\centering
	\includegraphics[scale=0.5]{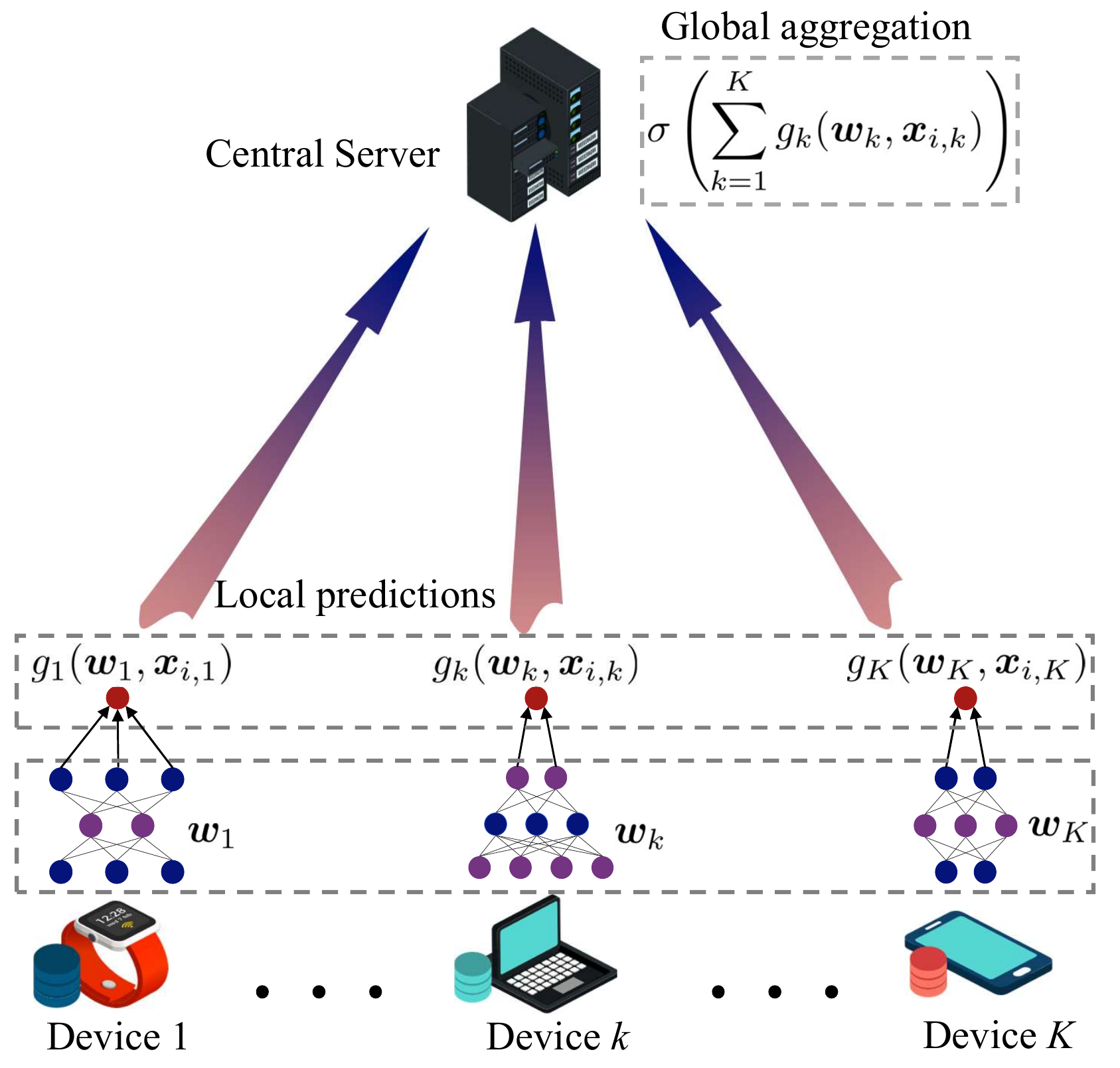}
	\vspace{-0.4cm}
	\caption{Aggregation of local prediction results for the $i$-th training sample in vertical FL.}
	\label{fig: vertical FL}
\end{figure}

% In this paper, we consider vertical FL with multiple types of devices geographically distributed over wireless networks.
% Different from the horizontal FL settings where the local datasets share the same feature space, we consider vertically partitioned datasets at different types of devices, i.e., different types of devices hold different features of the same samples.
% The powerful sensing and processing capabilities of these devices enable them to capture information from various perspectives (e.g., images, audio, and physical features), which provides feature diversity.
% One specific example is the Internet of Vehicles, which consists of many on-vehicle sensors and road condition sensors, such as on-vehicle cameras, GPS, inertial measurement units, surveillance cameras, vehicle mounted sensors and so on \cite{DBLP:journals/iotj/LiangYL19, DBLP:journals/icl/LiuZJLXC22}.
% These sensors of different types can collect data in different feature space to estimate the location of a vehicle or to discriminate the traffic conditions.
% We assume that different devices have non-overlapped features \cite{DBLP:journals/iotj/SuL21}.

Consider a vertical FL system consisting of $K$ IoT devices  and one central server, as shown in Fig. \ref{fig: vertical FL}.
Let $\mathcal{D} = \{ (\bm{x}_{i,1},\ldots,\bm{x}_{i,k}), y_i \}_{i=1}^{L}$ denote the whole training dataset of $L$ samples, where $\bm{x}_{i,k} \in \R^{d_k}$ denotes the partial  non-overlapped features with dimension $d_k$ of sample $i$ located at device $k$, and $y_i$ denotes the corresponding label only available  at the central server \cite{DBLP:journals/iotj/SuL21}.
% In vertical FL, it is assumed that the central server holds all labels $\mathcal{Y} = \{ y^i\}_{i=1}^{L}$, and device $k$ is only aware of its own local feature set $\mathcal{D}_k = \{ \bm{x}_{i,k}\}_{i=1}^{L} $ \cite{DBLP:journals/iotj/SuL21}. 
We use the concatenated vector $\bm{x}_i = [(\bm{x}_{i,1})^\mathsf{T},\ldots,(\bm{x}_{i,k})^\mathsf{T}]^\mathsf{T} \in \R^d$ to denote the whole feature vector of sample $i$ with dimension $d = \sum_{k=1}^{K} d_k$.
% Furthermore, each feature vector is assumed to be normalized, i.e., $\|\bm{x}_i\|^2 = \sum_{k=1}^K \|\bm{x}_{i,k}\|^2 = 1$.
The goal of vertical FL is to collaboratively learn a global model $\bm{w}$ that maps an input $\bm{x}_i$ to the corresponding prediction through a continuously differentiable function $h(\bm{w}, \bm{x}_i)$.
%We conclude the steps of vertical FL for devices and central server as follows.
% As the features of sample $i$ are distributed at different devices, 
We assume that device $k$ maps the local feature $\bm{x}_{i,k}$ to the local prediction result $g_k(\bm{w}_k, \bm{x}_{i,k})$ using sub-model $\bm{w}_k$, where $g_k(\cdot)$ is defined as the local prediction function.
The dimension of each local prediction result is the same, which is determined by the number of classes for the classification task. 
For example, for a binary classification task, the local prediction result $g_k(\bm{w}_k, \bm{x}_{i,k})$ is a scalar.

The input of the final prediction function $h(\bm{w}, \bm{x}_i)$ is based on the aggregation of the local prediction results $\{g_k(\bm{w}_k, \bm{x}_{i,k})\}$.
% Intuitively, even if different sensors hold different feature spaces, the scores for feature matching of the same class can be added together.
% Furthermore, element-wise sum allows the use of secure aggregation methods \cite{DBLP:journals/corr/BonawitzIKMMPRS16}, which can enhance the privacy and security of the algorithm.
We assume the final prediction is obtained by aggregating the local prediction results\footnote{For a wide range of models such as logistic classification and support vector machines
% \cite{9463409, abs-1912-11187, DBLP:journals/tsp/YingYS19}, 
the intermediate outputs of these models are additive, i.e., the final inner product can be split into the sum of multiple inner products.},
followed by a non-linear transformation $\sigma(\cdot)$ (e.g., a softmax function) as follows:
\begin{align}\label{eq: final prediction}
	h(\bm{w}, \bm{x}_i) = \sigma\left(\sum_{k=1}^K g_k(\bm{w}_k, \bm{x}_{i,k})\right),
\end{align}
where $\bm{w}$ is the global model concatenated by the sub-models $\{\bm{w}_k\}$, i.e., $\bm{w} = \left[(\bm{w}^\mathsf{T}_1, \cdots, \bm{w}^\mathsf{T}_K)\right]^\mathsf{T}$.
Based on \eqref{eq: final prediction}, the central server can aggregate local prediction results without directly accessing local features and local model.

To learn the global model $\bm{w}$, we focus on the following empirical risk minimization problem:
\begin{align}\label{pro:p1}
\min_{  \bm{w} \in\mathbb{R}^{d}} F(\bm{w}) &= \frac{1}{L} \sum_{i=1}^{L} f_i(\bm{w}) + \lambda \sum_{k=1}^K r_k(\bm{w}),
\end{align}
where $f_i(\cdot)$ is the sample-wise loss function indicating the loss between the prediction $h(\bm{w}, \bm{x}_i)$ and the true label $y_i$ for sample $i$, $r_k(\cdot)$ is the regularization for the sub-model $\bm{w}_k$ on device $k$, and $\lambda$ is a hyperparameter.

% In this paper, we consider using the full-batch gradient descent (GD) algorithm rather than the stochastic gradient descent (SGD) algorithm to solve problem \eqref{pro:p1}.
% There are two main reasons to adopt full-batch GD algorithm.
% First, the full-batch GD algorithm can make full use of the training dataset, thereby speeding up the convergence and reducing the total number of communication rounds.
% Second, the SGD algorithm requires batches of training samples to be synchronized between the devices and central server, which is however beyond the scope of this paper.
% In the following, we shall introduce the full-batch GD algorithm for vertical FL with perfect communication, which is also referred to as error-free gradient-based algorithm.

% In this paper, we consider using the full-batch gradient descent (GD) algorithm rather than the stochastic gradient descent (SGD) algorithm to simplify the analysis since the main goal of this paper is to design a general communication system for vertical FL.
% The SGD algorithm can be implemented by synchronizing batches of training samples between the devices and central server.
% In addition, using the SGD algorithm may slightly change the convergence analysis of this paper, but the overall framework can be applied.
% In the following, we shall introduce the full-batch GD algorithm for vertical FL with perfect communication, which is also referred to as error-free gradient-based algorithm.
In this paper, we consider using the full-batch gradient descent (GD) algorithm to solve \eqref{pro:p1}.
Let $\nabla F(\bm{w})$ denote the gradient of $F(\cdot)$ with respect to $\bm{w}$, then we obtain
\begin{equation}\label{eq: gradient}
	\nabla F(\bm{w}) = \frac{1}{L}\sum_{i=1}^L \begin{bmatrix}
		\nabla_1 f_i(\bm{w}) \\
		\vdots \\
		\nabla_K f_i(\bm{w})
	\end{bmatrix} + \lambda \begin{bmatrix}
		\nabla r_1(\bm{w})\\
		\vdots \\
		\nabla r_K(\bm{w})
	\end{bmatrix},
\end{equation}
where $\nabla_k f_i(\bm{w}) = \frac{\partial f_i(\bm{w})}{\partial \bm{w}_k}$ is the partial gradient of the loss function $f_i(\cdot)$ with respect to the sub-model $\bm{w}_k$.
Based on the chain rule in calculus, the partial gradient of $\nabla_k f_i(\bm{w})$ can be rewritten as
\begin{equation}\label{eq: partial gradient}
\begin{aligned}
	\nabla_k f_i(\bm{w}) &= f'_i\left( \sigma\left(\sum_{k=1}^K g_k(\bm{w}_k, \bm{x}_{i,k})\right)\right) \sigma'\left(\sum_{k=1}^K g_k(\bm{w}_k, \bm{x}_{i,k})\right)\nabla g_{k,i}(\bm{w}_k) \\
	&= G_i\left(\sum_{k=1}^K g_k(\bm{w}_k, \bm{x}_{i,k})\right) \nabla g_{k,i}(\bm{w}_k),
\end{aligned}
\end{equation}
where $f'_i(\cdot)$ and $\sigma'(\cdot)$ are the derivatives of $f_i(\cdot)$ and $\sigma(\cdot)$, $\nabla g_{k,i}(\bm{w}_k)=\frac{\partial g_k(\bm{w}_k, \bm{x}_{i,k})}{\partial \bm{w}_k}$ is the partial gradient of the local prediction function $g_k(\cdot)$ with respect to $\bm{w}_k$, and $G_i(\cdot)$ is an auxiliary function with respect to the aggregation of local prediction results for sample $i$.

In vertical FL, the features are distributed at local devices and the labels are known by the central server, so the partial gradients can be calculated separately.
% Specifically, since the computation of $\nabla g_{k,i}(\bm{w}_k)$ does not require the label information $y_i$,
% \footnote{Before each transmission of local prediction results, the server and all devices synchronize the indices of the training samples with low communication overhead, enabling the central server to aggregate local predictions of the same label.}
Specifically, device $k$ computes and uploads the local prediction result $g_k(\bm{w}_k, \bm{x}_{i,k})$ to the central server.
Then, the central server aggregates the local prediction results to obtain $\sum_{k=1}^K g_k(\bm{w}_k, \bm{x}_{i,k})$ and broadcast $G_i\left(\sum_{k=1}^K g_k(\bm{w}_k, \bm{x}_{i,k})\right)$ back to the devices.
Consequently, the devices obtain the partial gradient $\nabla_k f_i(\bm{w})$ through Eq. \eqref{eq: partial gradient} since $\nabla g_{k,i}(\bm{w}_k)$ can be calculated locally.
Each device then updates its local model by taking one step of gradient decent with learning rate $\mu$, i.e.,
\begin{align}\label{eq: gradient descent}
	\bm{w}^{(t+1)}_k = \bm{w}^{(t)}_k - \mu\left(\frac{1}{L}\sum_{i=1}^L\nabla_k f_i(\bm{w}^{(t)}) + \lambda \nabla r_k(\bm{w}^{(t)})\right).
\end{align}

Algorithm~\ref{algorithm 1} summarizes the above procedure.
Since the central server only receives the aggregation of local prediction results, i.e., neither local features nor local models are uploaded to the central server, the privacy of local data can be guaranteed.
Nevertheless, the number of local prediction results at each device is proportional to the number of training samples, which leads to high-dimensional data transmission.
In addition, all devices are required to participate the training process in vertical FL since each device owns a sub-model $\bm{w}_k$.
Hence, it is essential to design communication efficient technologies for vertical FL in wireless networks to tame the challenge brought by the high-dimensional data transmission with full device participation.

\begin{algorithm}[t]
	\setstretch{0.95}
	\caption{Error-free gradient-based algorithm for vertical FL} 
	\label{algorithm 1}
	  %   \DontPrintSemicolon
		  \SetAlgoLined
  %       \KwData{local feature sets $\{\mathcal{D}_k\}$}
		  \SetKwInOut{Require}{Require}
		  \SetKwInOut{Output}{Output}
		  \SetKwFor{ParFor}{}{do in parallel}{end}
		  \SetKwFor{Aggre}{}{do}{end}
		  \Require{learning rate $\mu$, maximum communication rounds $T$, initialization $\{\bm{w}_k^{(0)}\}$.}
		  \Output{$\bm{w}^{(T)} = \left[(\bm{w}_1^{(T)})^\mathsf{T}, \ldots, (\bm{w}_K^{(T)})^\mathsf{T}\right]^\mathsf{T}$.}
		  \For{$t=0,1,\ldots,T$}{
		  \ParFor{each device $k \in [K]$}{
		  \If{$t=0$}{
		  Send the local prediction results $\left\{g_k(\bm{w}^{(0)}_k, \bm{x}_{i,k})\right\}^L_{i=1}$ to the central server\;}
		  \Else{
		%   Receive $\left\{G_i\left(\sum_{k=1}^K g_k(\bm{w}^{(t)}_k, \bm{x}_{i,k})\right) \right\}^L_{i=1}$\;
		  Compute $\{\nabla_k f(\bm{w}^{(t)})\}_{i=1}^L$ via \eqref{eq: partial gradient}, and update the local model $\bm{w}^{(t+1)}_k$ via  \eqref{eq: gradient descent}\;
		  Send the local prediction results $\left\{g_k(\bm{w}^{(t+1)}_k, \bm{x}_{i,k})\right\}^L_{i=1}$ to the central server\;}
		  }
		  \Aggre{the central server}{Receive and aggregate the local prediction results $\left\{\sum_{k=1}^K g_k(\bm{w}^{(t)}_k, \bm{x}_{i,k})\right\}^L_{i=1}$\;
		  % Aggregate the local prediction results $(\bm{w}^{(t+1)})^\mathsf{T}\bm{x}_i,\quad \forall i \in [L]$ \;
		  Compute $\left\{G_i\left(\sum_{k=1}^K g_k(\bm{w}^{(t)}_k, \bm{x}_{i,k})\right) \right\}^L_{i=1}$ and broadcast it back to all the devices\;}
		  }
\end{algorithm}

\subsection{Over-the-Air Computation based Cloud Radio Access Network}
To enable fast aggregation with full device participation in vertical FL,
% Cloud-RAN is able to effectively mitigate the communication straggler issues and combat the unfavorable channel propagation.
we shall propose an AirComp based Cloud-RAN to support vertical FL over ultra dense wireless networks, where multiple RRHs serving as edge servers are coordinated by a centralized BBU pool serving as a central server, as shown in Fig. \ref{fig: transmission model}.
In particular, the central server coordinates $N$ edge servers with limited-capacity digital fronthaul links to serve $K$ single-antenna devices.
Each edge server is equipped with $M$ antennas.
Let $C_n$ denote the fronthaul capacity of the link between edge server $n$ and the central server.
The fronthaul capacities need to satisfy an overall capacity constraint, i.e., $\sum_{n=1}^NC_n \leq C$.
We assume perfect synchronization\footnote{
For the AirComp-based system, we can implement synchronization by sharing a reference-clock among the devices \cite{DBLP:conf/infocom/AbariRKP15}, or using the timing advance technique commonly adopted in 4G Long Term Evolution and 5G New Radio \cite{DBLP:journals/cm/MahmoodAGTS19}.} between among devices and perfect channel state information (CSI) between the edge servers to the devices is available to the central server and the devices\footnote{The RRHs estimates the channel state information (CSI) by receiving pilot signals sent from devices \cite{DBLP:journals/tcom/StephenZ18}.}.
The notations and parameters used in the system model are summarized in Table \ref{table1}.
% as in \cite{DBLP:journals/tsp/Zhou016, DBLP:journals/tcom/LiuBZ15} 
% \footnote{In practice, the CSI can be acquired through pilot sequences sent by the devices \cite{DBLP:journals/tcom/StephenZ18}.}
In the following, we introduce the communication models for the uplink and downlink transmission to vertically train the FL model, respectively.

\begin{table}[]
	\centering
	\caption{\change{Important Notations}}
		\begin{tabular}{@{}ll@{}}
			\toprule
			Notation	&  Definition	\\ \midrule
			$K$			& Number of edge devices \\
			$N$ 		&  Number of edge servers \\
			$M$			&  Number of antennas for each edge server \\
			$C_n$		&  Fronthaul capacity of edge server $n$ \\ 
			$L$ & Number of time slots (samples)\\ 
			$\bm{h}^{(t)}_{\mathrm{UL},k,n}$ & Uplink channel between device $k$ and server $n$ \\
			$b^{(t)}_{\mathrm{UL},k}$ & Uplink transmit scalar for device $k$ in the $t$-th round \\
			$s^{(t)}_k(i)$	& Uplink transmit signal for device $k$ in the $t$-th round\\
			$z^{(t)}_{\mathrm{UL},n}(i)$ & Uplink channel noise for RRH $n$ in the $t$-th round\\
			$P_\mathrm{UL}$ & Maximum uplink transmit power\\
			$\bm{q}^{(t)}_{\mathrm{UL},n}(i)$ & Uplink quantization noise for RRH $n$ in the $t$-th round\\
			$\bm{Q}^{(t)}_{\mathrm{UL},n}$ & Covariance matrix of $\bm{q}^{(t)}_{\mathrm{UL},n}(i)$ \\
			$n^{(t)}_\mathrm{UL}(i)$ & Effective uplink noise in the $t$-th round\\
			$\bm{h}^{(t)}_{\mathrm{DL},k,n}$ & Downlink channel between device $k$ and server $n$ \\
			$z^{(t)}_{\mathrm{DL},n}(i)$ & Downlink channel noise for RRH $n$ in the $t$-th round\\
			$P_\mathrm{DL}$ & Maximum downlink transmit power\\
			$\bm{q}^{(t)}_{\mathrm{DL},n}(i)$ & Downlink quantization noise for RRH $n$ in the $t$-th round \\
			$\bm{Q}^{(t)}_{\mathrm{DL},n}$ & Covariance matrix of $\bm{q}^{(t)}_{\mathrm{DL},n}(i)$ \\
			$n^{(t)}_\mathrm{DL}(i)$ & Effective downlink noise in the $t$-th round\\
			\bottomrule
	\end{tabular}
	\label{table1}
\end{table}

\begin{figure}[t]
\centering
	\centering
	\includegraphics[scale=0.4]{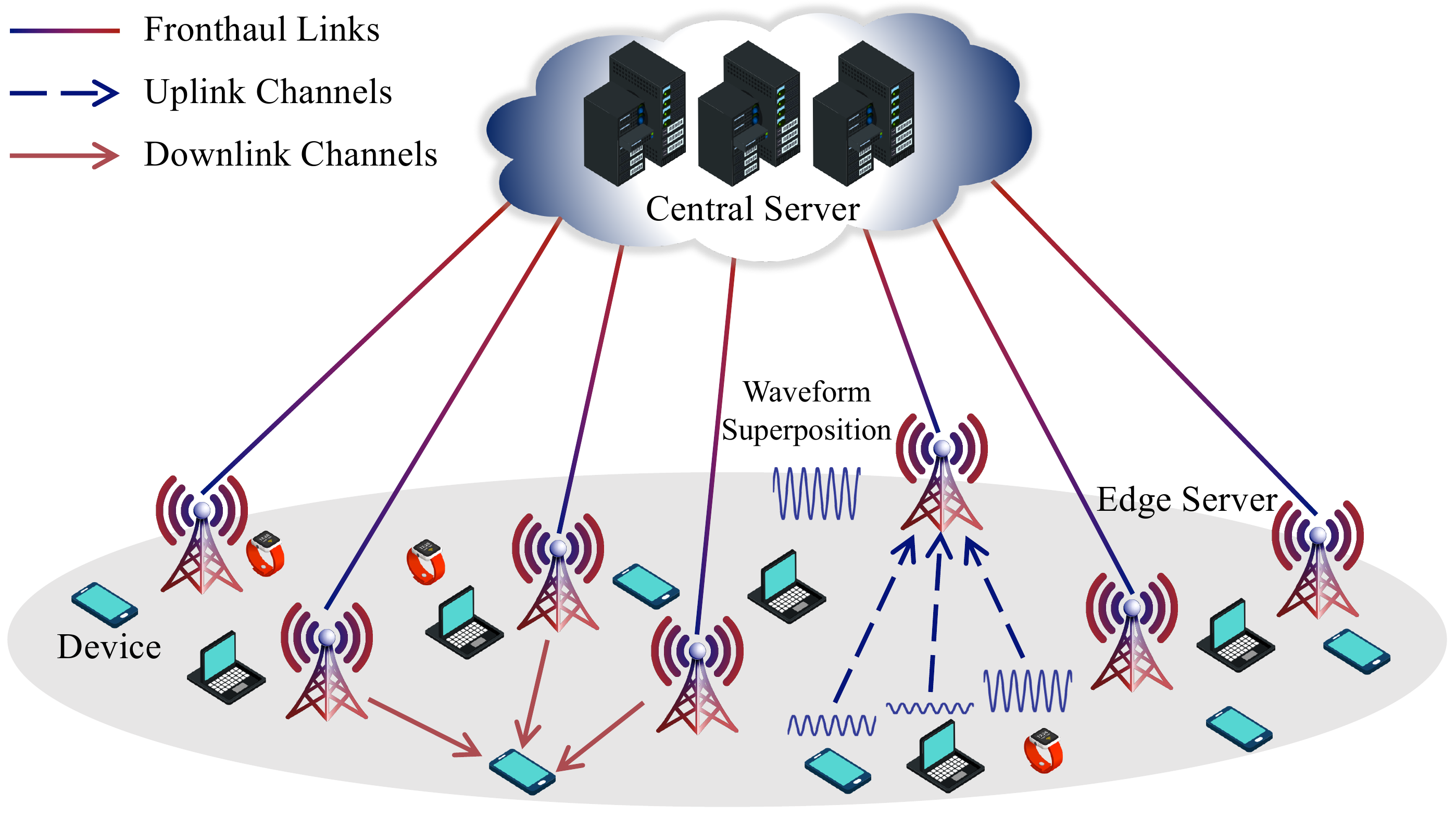}
	\caption{AirComp-based Cloud-RAN network for vertical FL.}
	\label{fig: transmission model}
\end{figure}

\subsubsection{Uplink Transmission Model} In the uplink transmission, we assume the devices communicate with the edge servers over a shared wireless multiple access channel via AirComp.
In this case, the edge servers first receive the AirComp results of the local prediction results transmitted by the devices, and then forward the aggregated results to the central server.
By exploiting the signal superposition property of a wireless multiple access channel, the central server directly processes the aggregated version of analog modulated local information simultaneously transmitted by the devices, which significantly reduces transmission latency.

The channels between the devices and the edge servers are assumed to be quasi-static flat-fading \cite{DBLP:journals/tsp/Zhou016, DBLP:journals/tsp/LiuZ15}, where the channel coefficients remain the same during one communication block. 
Each block is assumed to contain $L$ time slots, so that each device is able to transmit an $L$-dimensional signal vector encoding local prediction results of $L$ samples within one block\footnote{
For large data sample size, the local prediction results can be transmitted within multiple consecutive coherence blocks, which will be left in the future work.}.
We assume the devices complete the $t$-th communication round for model training in the $t$-th transmission block.
Then, the received signal at edge server $n$ during the $i$-th uplink transmission time slot in the $t$-th communication round is given by
\begin{align}\label{eq: MAC channel}
        \bm{y}^{(t)}_{\mathrm{UL},n}(i) = \sum_{k=1}^K\bm{h}^{(t)}_{\mathrm{UL}, k,n} b^{(t)}_{\mathrm{UL}, k}s^{(t)}_k(i)+ \bm{z}_{\mathrm{UL},n}^{(t)}(i),
\end{align}
where $\bm{h}^{(t)}_{\mathrm{UL}, k,n} \in \mathbb{C}^M$ is the uplink channel between device $k$ and edge server $n$ in the $t$-th communication round, $b^{(t)}_{\mathrm{UL}, k}$ denotes the transmit scalar,  $s^{(t)}_k(i)=g_k(\bm{w}^{(t)}_k, \bm{x}_{i,k})$ denotes the transmit signal, and $\bm{z}_{\mathrm{UL},n}^{(t)}(i) \sim \mathcal{CN}(0, \sigma^2_z)$ denotes the additive white Gaussian noise (AWGN) for edge server $n$.
%Note that each device uses the same transmit scalar in one communication round since the channels are assumed to be invariant during one transmission block. 
Without loss of generality, we assume that the transmit signals $\{s^{(t)}_k(i)\}$ follow the standard Gaussian distribution\footnote{Each neuron in a neural network receives multiple input signals that vary randomly during both forward and backward propagation. Due to the central limit theorem, the sum of these random variables tends to be Gaussian distributed.}, i.e., $s^{(t)}_k(i) \sim \mathcal{N}(0,1)$.
% \cite{DBLP:journals/twc/ZhuDGH21}.
Hence, the transmit power constraint of each device is written as
\begin{align}\label{eq: uplink power constraint}
        \mathbb{E}\left[|b^{(t)}_{\mathrm{UL}, k}s^{(t)}_k(i) |^2\right]  &= |b^{(t)}_{\mathrm{UL}, k}|^2\mathbb{E}[(s^{(t)}_k(i))^2] \leq \left|b^{(t)}_{\mathrm{UL}, k}\right|^2\leq P_{\mathrm{UL}},\quad \forall k \in [K],
\end{align}
where $P_{\mathrm{UL}}$ is the maximum transmit power for each device.

To forward the received signals to the central server through the limited fronthaul links, edge server $n$ quantizes the received signal $\bm{y}^{(t)}_{\mathrm{UL},n}(i)$ and then sends the quantized signals to the central server. 
In this paper, we use compression-based strategies to transmit the information between the central server and edge servers via fronthaul links.
By adopting Gaussian quantization test channel\footnote{To model the cost for transmitting high fidelity version of $\bm{y}^{(t)}_{\mathrm{UL},n}(i)$, we study a theoretical quantization model by viewing Eq. \eqref{eq: gaussian test} as a test channel based on the rate-distortion \cite{DBLP:journals/spm/ParkSSS14}.
%  and derive the corresponding convergence based on the rate-distortion  \cite{DBLP:journals/spm/ParkSSS14}, which can serve as a performance upper bound.
The results can be extended to the other quantization schemes, i.e.,  uniform scalar quantization schemes.},
the quantized signal $\tilde {\bm{y}}^{(t)}_{\mathrm{UL},n}(i)$ received at the central server from edge server $n$ can be modeled as \cite{DBLP:journals/spm/ParkSSS14}
\begin{align}\label{eq: gaussian test}
	\tilde {\bm{y}}^{(t)}_{\mathrm{UL},n}(i) = \bm{y}^{(t)}_{\mathrm{UL},n}(i) + \bm{q}^{(t)}_{\mathrm{UL}, n}(i),
\end{align}
where $\bm{q}^{(t)}_{\mathrm{UL}, n}(i) \in \mathbb{C}^M \sim \mathcal{CN}\left(\bm{0}, \bm{Q}^{(t)}_{\mathrm{UL}, n}\right)$ denotes the quantization noise and $\bm{Q}^{(t)}_{\mathrm{UL}, n}$ is the diagonal covariance matrix of the quantization noise for edge server $n$.
% The rate-distortion can be easily understood in terms of the quantization noise $\bm{q}^{(t)}_{\mathrm{UL}, n}(i)$ introduced in Eq. (8).
The quantization noise level directly provides an indication of the accuracy of $\tilde {\bm{y}}^{(t)}_{\mathrm{UL},n}(i)$.
On the other hand, the level of the quantization noise indicates the amount of fronthaul capacity required for compression. 
In general, higher fronthaul capacity results in better compression resolution and less quantization noise.

% Based on the rate-distortion theory, the fronthaul rates of $N$ edge servers at the $i$-th time slot in the $t$-th communication round should satisfy \cite[Ch. 3]{DBLP:books/cu/GK2011}
Based on the rate-distortion theory, the fronthaul rates of $N$ edge servers at the $i$-th time slot in the $t$-th communication round should satisfy \cite[Ch. 3]{DBLP:books/cu/GK2011}
\begin{align}\label{eq: uplink capacity constraint}
	&\sum_{n=1}^N C^{(t)}_{\mathrm{UL}, n} = \sum_{n=1}^N I\left(\bm{y}^{(t)}_{\mathrm{UL},n}(i);\tilde {\bm{y}}^{(t)}_{\mathrm{UL},n}(i)\right) \nonumber \\
	=& \sum_{n=1}^N\log\frac{\left|\sum_{k=1}^{K} |b^{(t)}_{\mathrm{UL}, k}|^2\bm{h}^{(t)}_{\mathrm{UL}, k, n} (\bm{h}^{(t)}_{\mathrm{UL}, k, n})^\mathsf{H}+\sigma^2_z\bm{I}+\bm{Q}^{(t)}_{\mathrm{UL}, n}\right|}{\left|\bm{Q}^{(t)}_{\mathrm{UL}, n}\right|}\nonumber \\
	% \leq &\sum_{n=1}^N\log\frac{\left|P_{\mathrm{UL}}\sum_{k=1}^{K} \bm{h}^{(t)}_{\mathrm{UL}, k, n} (\bm{h}^{(t)}_{\mathrm{UL}, k, n})^\mathsf{H}+\sigma^2_z\bm{I}+\bm{Q}^{(t)}_{\mathrm{UL}, n}\right|}{\left|\bm{Q}^{(t)}_{\mathrm{UL}, n}\right|}\nonumber \\
	\leq&\log\frac{\left|P_{\mathrm{UL}}\sum_{k=1}^{K} \bm{h}^{(t)}_{\mathrm{UL}, k} (\bm{h}^{(t)}_{\mathrm{UL}, k})^\mathsf{H}+\sigma^2_z\bm{I}+\bm{Q}^{(t)}_{\mathrm{UL}}\right|}{\left|\bm{Q}^{(t)}_{\mathrm{UL}}\right|}
	\leq C,
\end{align}
where $I(X;Y)$ denotes the mutual information between input $X$ and output $Y$, 
% \begin{align}
% \bm{h}^{(t)}_{\mathrm{UL}, k} = \left[\left(\bm{h}_{\mathrm{UL},k,1}^{(t)}\right)^{\mathsf{H}}, \ldots, \left(\bm{h}_{\mathrm{UL},k,N}^{(t)}\right)^{\mathsf{H}}\right]^\mathsf{H} \in \mathbb{C}^{NM},
% \end{align}
$\bm{h}^{(t)}_{\mathrm{UL}, k}$ is the concatenated channel vector of $\{\bm{h}_{\mathrm{UL},k,1}^{(t)}\}$, and $\bm{Q}^{(t)}_\mathrm{UL}$ is the uplink covariance matrix defined by $\bm{Q}^{(t)}_\mathrm{UL}=\text{diag}\left\{\bm{Q}^{(t)}_{\mathrm{UL}, 1},\ldots,\bm{Q}^{(t)}_{\mathrm{UL}, N}\right\}$.

Stacking the quantized signals $\{\tilde{\bm{y}}^{(t)}_\mathrm{UL}(i)\}$ transmitted from all edge servers yields
\begin{align*}
	\tilde {\bm{y}}^{(t)}_{\mathrm{UL}}(i) 
	% &= \left[\left(\tilde{\bm{y}}^{(t)}_{\mathrm{UL},1}(i)\right)^\mathsf{H}, \ldots, \left(\tilde{\bm{y}}^{(t)}_{\mathrm{UL},N}(i)\right)^\mathsf{H}\right]^\mathsf{H} \nonumber \\
	= \sum_{k=1}^K\bm{h}^{(t)}_{\mathrm{UL}, k} b^{(t)}_{\mathrm{UL}, k}s^{(t)}_k(i) + \bm{z}_{\mathrm{UL}}^{(t)}(i) + \bm{q}^{(t)}_{\mathrm{UL}}(i),
\end{align*}
where
\begin{align*}
	\bm{z}_{\mathrm{UL}}^{(t)}(i) = \left[\left(\bm{z}_{\mathrm{UL},1}^{(t)}(i)\right)^{\mathsf{H}}, \ldots, \left(\bm{z}_{\mathrm{UL},N}^{(t)}(i)\right)^{\mathsf{H}}\right]^\mathsf{H}, 
	\bm{q}^{(t)}_{\mathrm{UL}}(i) = \left[\left(\bm{q}^{(t)}_{\mathrm{UL}, 1}(i)\right)^{\mathsf{H}}, \ldots, \left(\bm{q}^{(t)}_{\mathrm{UL}, N}(i)\right)^{\mathsf{H}}\right]^\mathsf{H}.
\end{align*}
%where 
%\begin{align*}
%	\bm{z}_{\mathrm{UL}}^{(t)}(i) &= \left[\left(\bm{z}_{\mathrm{UL},1}^{(t)}(i)\right)^{\mathsf{H}}, \ldots, \left(\bm{z}_{\mathrm{UL},N}^{(t)}(i)\right)^{\mathsf{H}}\right]^\mathsf{H} \in \mathbb{C}^{NM}, \\
%	\bm{q}^{(t)}_{\mathrm{UL}}(i) &= \left[\left(\bm{q}^{(t)}_{\mathrm{UL}, 1}(i)\right)^{\mathsf{H}}, \ldots, \left(\bm{q}^{(t)}_{\mathrm{UL}, N}(i)\right)^{\mathsf{H}}\right]^\mathsf{H} \in \mathbb{C}^{NM}.
%\end{align*}
We design receive beamforming vectors at the central server for $\tilde {\bm{y}}^{(t)}_{\mathrm{UL}}(i)$, and take the real part to estimate the aggregation of local predictions $s^{(t)}(i) = \sum_{k=1}^Kg_k(\bm{x}_{i,k};\bm{w}^{(t)}_k)$, through the following procedure 
\begin{align}\label{eq: estimate}
	\hat{s}^{(t)}(i) &= \frac{1}{\sqrt{\eta^{(t)}}} \Re((\bm{m}^{(t)})^\mathsf{H} \tilde{\bm{y}}^{(t)}_\mathrm{UL}(i)) = \frac{1}{\sqrt{\eta^{(t)}}} \Re\left((\bm{m}^{(t)})^\mathsf{H} \sum_{k=1}^K\bm{h}^{(t)}_{\mathrm{UL}, k} b^{(t)}_{\mathrm{UL}, k}s^{(t)}_k(i)\right) +  n^{(t)}_\mathrm{UL}(i),
\end{align}
where $n^{(t)}_\mathrm{UL}(i) = \frac{1}{\sqrt{\eta^{(t)}}} \Re\left((\bm{m}^{(t)})^\mathsf{H}\left(\bm{z}_{\mathrm{UL}}^{(t)}(i) + \bm{q}^{(t)}_{\mathrm{UL}}\right)\right)$ is the effective uplink noise, and $\bm{m}^{(t)} = [(\bm{m}^{(t)}_1)^\mathsf{H}, \ldots, (\bm{m}^{(t)}_N)^\mathsf{H}]^\mathsf{H}\in\mathbb{C}^{NM}$ is the receive beamforming vector and $\eta^{(t)}$ is the power control factor.
Given $\bm{m}^{(t)}$ and $\eta^{(t)}$, the effective uplink noise is distributed as $n^{(t)}_\mathrm{UL}(i) \sim \mathcal{N}(0, \sigma^2_{\mathrm{UL}}(t))$ with the variance 
\begin{align}\label{eq: uplink noise}
	\sigma^2_{\mathrm{UL}}(t) = \frac{1}{2\eta^{(t)}}(\bm{m}^{(t)})^\mathsf{H}\left(\sigma^2_z\bm{I}+\bm{Q}^{(t)}_\mathrm{UL}\right)\left(\bm{m}^{(t)}\right).
\end{align}

\subsubsection{Downlink Transmission Model} After obtaining the estimate $\hat{s}^{(t)}(i)$, the central server computes $G_i\left(\hat{s}^{(t)}(i)\right)$ based on the noisy aggregation $\hat{s}^{(t)}(i)$ and then broadcasts the result back to the devices.
% For the downlink communication, we adopt the compression-based strategy, where
The central server first computes the beamformed signals to encode $G_i\left(\hat{s}^{(t)}(i)\right)$, and then send them to the edge servers over the limited fronthaul links in the $t$-th communication round.
The edge servers directly broadcast the beamformed signals to the devices over the wireless fading channels.
The resulting transmit signal at the central server for edge server $n$ is given as
$
\tilde{\bm{r}}^{(t)}_{\mathrm{DL}, n}(i) = \bm{u}^{(t)}_nG_i\left(\hat{s}^{(t)}(i)\right),
$
where $ \bm{u}^{(t)}_n  \in \C^{M}$ is the transmit beamforming vector at edge server $n$.

Similar to the uplink transmission, these signals are then compressed and sent to the edge servers over the finite-capacity fronthaul link.
The compression noise is modeled as 
\begin{align}
    \bm{r}^{(t)}_{\mathrm{DL}, n}(i) = \tilde{\bm{r}}^{(t)}_{\mathrm{DL}, n}(i) + \bm{q}^{(t)}_{\mathrm{DL}, n}(i),\quad \forall i \in [L],
\end{align} 
where $\bm{r}^{(t)}_{\mathrm{DL}, n}(i)$ is the reconstructed signal that edge server $n$ actually broadcasts to the devices, and $\bm{q}_{\mathrm{DL}, n}^{(t)}(i) \in \C^M \sim \mathcal{CN}(\bm{0}, \bm{Q}^{(t)}_{\mathrm{DL}, n})$ denotes the downlink quantization noise with the covariance matrix $\bm{Q}^{(t)}_{\mathrm{DL}, n}$.
%Assuming that $G_i\left(\hat{s}^{(t)}(i)\right)$ is upper bounded by a constant $\gamma$, the transmit power at edge server $n$ should satisfy

Without loss of generality, we assume that the transmit signal follows the standard Gaussian distribution, i.e., $G_i\left(\hat{s}^{(t)}(i)\right)\sim\mathcal{N}(0,1)$, the transmit power for all edge servers should satisfy
\begin{align}\label{eq: downlink power constraint}
    \sum_{n=1}^N\mathbb{E}\left[\|\bm{r}^{(t)}_{\mathrm{DL}, n}(i)\|^2\right] &= \sum_{n=1}^N\mathbb{E}\left[\|\bm{u}^{(t)}_nG_i\left(\hat{s}^{(t)}(i)\right)\|^2\right] + \sum_{n=1}^N\trace(\bm{Q}^{(t)}_{\mathrm{DL}, n})  \nonumber \\
    &= \|\bm{u}^{(t)}\|^2 + \trace(\bm{Q}^{(t)}_{\mathrm{DL}})   \leq P_{\mathrm{DL}}, 
\end{align}
where $P_{\mathrm{DL}}$ is the maximum total transmit power for all edge servers, $\bm{u}^{(t)} = [(\bm{u}^{(t)}_1)^\mathsf{H}, \ldots, (\bm{u}^{(t)}_N)^\mathsf{H}]^\mathsf{H}$ is the concatenated transmit beamforming vector, and $\bm{Q}^{(t)}_\mathrm{DL}$ is the downlink covariance matrix defined by $\bm{Q}^{(t)}_\mathrm{DL}=\text{diag}\{\bm{Q}^{(t)}_{\mathrm{DL}, 1},\ldots,\bm{Q}^{(t)}_{\mathrm{DL}, N}\}$.
Similarly, the fronthaul capacity required for independent compression should satisfy
% Similarly, the fronthaul capacity required for independent compression should satisfy \cite{quek_peng_simeone_yu_2017}
\begin{align}\label{eq: downlink capacity constraint}
    % &\sum_{n=1}^N C^{(t)}_{\mathrm{DL}, n} = \sum_{n=1}^N I (\bm{r}^{(t)}_{\mathrm{DL}, n}(i); \tilde{\bm{r}}^{(t)}_{\mathrm{DL}, n}(i)) \nonumber \\
    % =& \sum_{n=1}^N \log \frac{\left|\bm{u}^{(t)}_nG_i\left(\hat{s}^{(t)}(i)\right)(\bm{u}^{(t)}_nG_i\left(\hat{s}^{(t)}(i)\right))^\mathsf{H} + \bm{Q}^{(t)}_{\mathrm{DL}, n}\right|}{\left|\bm{Q}^{(t)}_{\mathrm{DL}, n}\right|} \nonumber \\
    % \leq& \sum_{n=1}^N  \log \frac{\left|\bm{u}^{(t)}_n(\bm{u}^{(t)}_n)^\mathsf{H} + \bm{Q}^{(t)}_{\mathrm{DL}, n}\right|}{\left|\bm{Q}^{(t)}_{\mathrm{DL}, n}\right|} = \log \frac{\left|\bm{u}^{(t)}(\bm{u}^{(t)})^\mathsf{H} + \bm{Q}^{(t)}_{\mathrm{DL}}\right|}{\left|\bm{Q}^{(t)}_{\mathrm{DL}}\right|} \leq C.
	&\sum_{n=1}^N C^{(t)}_{\mathrm{DL}, n} = \sum_{n=1}^N I (\bm{r}^{(t)}_{\mathrm{DL}, n}(i); \tilde{\bm{r}}^{(t)}_{\mathrm{DL}, n}(i))
	= \log \frac{\left|\bm{u}^{(t)}(\bm{u}^{(t)})^\mathsf{H} + \bm{Q}^{(t)}_{\mathrm{DL}}\right|}{\left|\bm{Q}^{(t)}_{\mathrm{DL}}\right|} \leq C.
\end{align}

The received signal at device $k$ under the compression strategy can be expressed as 
\begin{align}
    &y^{(t)}_{\mathrm{DL}, k}(i) = \sum_{n=1}^N \left(\bm{h}^{(t)}_{\mathrm{DL},k,n}\right)^\mathsf{H}\bm{r}^{(t)}_{\mathrm{DL}, n}(i) + z^{(t)}_{\mathrm{DL}, k}(i) \nonumber \\\notag
    % &= G_i\left(\hat{s}^{(t)}(i)\right)\sum_{n=1}^N \left(\bm{h}^{(t)}_{\mathrm{DL},k,n}\right)^\mathsf{H}\bm{u}^{(t)}_n  + \sum_{n=1}^N \left(\bm{h}^{(t)}_{\mathrm{DL},k,n}\right)^\mathsf{H}\bm{q}^{(t)}_{\mathrm{DL}, n} + z^{(t)}_{\mathrm{DL}, k}(i) \nonumber \\
    &= G_i\left(\hat{s}^{(t)}(i)\right) (\bm{h}^{(t)}_{\mathrm{DL}, k})^\mathsf{H} \bm{u}^{(t)} + (\bm{h}^{(t)}_{\mathrm{DL}, k})^\mathsf{H} \bm{q}^{(t)}_\mathrm{DL} + z^{(t)}_{\mathrm{DL}, k}(i),
\end{align}
where 
\begin{align*}
	\bm{Q}^{(t)}_{\mathrm{DL}} = \left[(\bm{q}^{(t)}_{\mathrm{DL},1})^\mathsf{H}, \ldots, (\bm{q}^{(t)}_{\mathrm{DL}, N})^\mathsf{H}\right]^\mathsf{H},
	\bm{h}^{(t)}_{\mathrm{DL},k} = \left[(\bm{h}^{(t)}_{\mathrm{DL},k,1})^\mathsf{H}, \ldots, (\bm{h}^{(t)}_{\mathrm{DL},k,N})^\mathsf{H}\right]^\mathsf{H}.
\end{align*}
% and the downlink AWGN $z^{(t)}_{\mathrm{DL}, k}(i)\sim\mathcal{CN}(0, \sigma^2_z)$.

After receiving signals, device $k$ estimates $G_i\left(\hat{s}^{(t)}(i)\right)$ by scaling the received signals with the receive scalar $b^{(t)}_{\mathrm{DL},k}$ and taking the real part, i.e., 
\begin{align}\label{eq: estimate G}
    \hat{G}_{i,k}\left(\hat{s}^{(t)}(i)\right)= \Re\left(b^{(t)}_{\mathrm{DL}, k} y^{(t)}_{\mathrm{DL}, k}(i)\right) = \Re\left(b^{(t)}_{\mathrm{DL}, k} G_i\left(\hat{s}^{(t)}(i)\right) (\bm{h}^{(t)}_{\mathrm{DL}, k})^\mathsf{H} \bm{u}^{(t)}\right) + n^{(t)}_{\mathrm{DL},k}(i),
\end{align}
where $\hat{G}_{i,k}\left(\hat{s}^{(t)}(i)\right)$ is the estimate of $G_i\left(\hat{s}^{(t)}(i)\right)$ at device $k$, and 
\begin{align}
	n^{(t)}_{\mathrm{DL}, k}(i) = \Re\left(b^{(t)}_{\mathrm{DL}, k}((\bm{h}^{(t)}_{\mathrm{DL}, k})^\mathsf{H} \bm{q}^{(t)}_\mathrm{DL} + z^{(t)}_{\mathrm{DL}, k}(i))\right)
\end{align}
denotes the effective downlink noise for device $k$ distributed as $n^{(t)}_{\mathrm{DL}, k}(i) \sim \mathcal{N}(0, \sigma^2_{\mathrm{DL},k}(t))$.
The variance is given by
\begin{align}\label{eq: downlink noise}
    \sigma^2_{\mathrm{DL},k}(t) = \frac{|b^{(t)}_{\mathrm{DL}, k}|^2}{2}\left(\sigma_{z}^2+(\bm{h}^{(t)}_{\mathrm{DL}, k})^\mathsf{H} \bm{Q}^{(t)}_\mathrm{DL} \bm{h}^{(t)}_{\mathrm{DL}, k}\right).
\end{align}

The proposed framework performs centralized signal processing via AirComp by exploiting the cooperation among edge servers, thereby achieving reliable and low-latency model aggregation.
Note that we assume all the RRHs are active during FL training, since vertical FL require all the devices to participate in FL training.
A key challenge of the Cloud-RAN system is that the quantization errors caused by the limited fronthaul capacity and the aggregation errors caused by AirComp will severely degrade the learning performance of vertical FL. 
In this paper, our goal is to understand the impact of limited fronthaul capacity and AirComp aggregation error on the learning performance for vertical FL, followed by designing efficient Cloud-RAN system optimization framework to improve the vertical FL performance.
To this end, we analyze the convergence of Algorithm \ref{algorithm 1} under the AirComp based Cloud-RAN, and characterize the convergence with respect to the transmission parameters (e.g., transmit/receive beamforming vectors and quantization covariance matrices) in the next section.

\section{Convergence Analysis}
In this section, we characterize the convergence behavior of Algorithm \ref{algorithm 1} under the AirComp based Cloud-RAN with limited fronthaul capacity.
Specifically, we first derive the closed-form expression of the vertical FL convergence in terms of optimality gap between objective values and the optimal value $F(\bm{w}^*)$, i.e., $\mathbb{E}[F(\bm{w}^{(T)})] - F(\bm{w}^*)$, where the expectation is taken over the communication noise.
% We define the optimality gap of the objective value with respect to the optimal value $F(\bm{w}^*)$ in the $T$-th communication round, i.e., $\mathbb{E}[F(\bm{w}^{(T)})] - F(\bm{w}^*)$, where the expectation is taken over the communication noise.
% Based on the derived expression of the optimality gap, we then design the communication system by jointly optimizing the uplink and downlink system parameters, e.g., $\{\bm{m}^{(t)}, \bm{u}^{(t)}$, $\bm{Q}^{(t)}_{\mathrm{UL}}$ and $\bm{Q}^{(t)}_{\mathrm{DL}}\}$, which will be presented in Section \ref{sec: system optimization}.

\subsection{Zero-forcing Precoding}
In this subsection, we first design the transmit and receiver scalars of devices to compensate the channel fading based on perfect CSI. 
% For the uplink transmission, each device exploits the power control factor $\eta^{(t)}$, receive beamforming vector $\bm{m}^{(t)}$ and known wireless channel gain $\bm{h}^{(t)}_{\mathrm{UL},k}$ for channel inversion.
To achieve channel inversion, the transmit scalar $b^{(t)}_{\mathrm{UL}, k}$ at device $k$ is designed as
	\begin{align}\label{eq: transmit scalar}
		b^{(t)}_{\mathrm{UL}, k} = \frac{\sqrt{\eta^{(t)}}((\bm{m}^{(t)})^\mathsf{H}\bm{h}^{(t)}_{\mathrm{UL},k})^\mathsf{H}}{|(\bm{m}^{(t)})^\mathsf{H}\bm{h}^{(t)}_{\mathrm{UL},k}|^2}, \quad \forall k \in [K].
	\end{align}
Here, we assume that $\bm{m}^{(t)}$ and $\eta^{(t)}$ are known at devices through the feedback process before transmission.
% and that they were synchronized \cite{DBLP:journals/twc/AmiriGKP22}.
As a result, the received signal at the central server becomes the desired aggregation of the local prediction results with noises.
Recall \eqref{eq: estimate}, the estimate $\hat{s}^{(t)}(i)$ at the central server can be written as 
\begin{align}\label{eq: estimate s}
	\hat{s}^{(t)}(i) = \sum_{k=1}^Kg_k(\bm{w}^{(t)}_k, \bm{x}_{i,k}) +n^{(t)}_\mathrm{UL}(i) = s^{(t)}(i)  + n^{(t)}_\mathrm{UL}(i),
\end{align}
for all $i \in [L]$.
Since the uplink noise $n^{(t)}_{\mathrm{UL}}(i)$ is zero-mean and independent to $s^{(t)}(i)$, the estimate $\hat{s}^{(t)}(i)$ is an unbiased estimation of $s^{(t)}(i)$, i.e., $\mathbb{E}[\hat{s}^{(t)}(i)] = s^{(t)}(i)$, with the expectation  taken over the uplink noise.
Considering the transmit power constraint of each device, i.e., $\|b^{(t)}_{\mathrm{UL},k}\|^2 \leq P_{\mathrm{UL}}$, the power control factor is given by $\eta^{(t)} = P_{\mathrm{UL}}\min_k \left|(\bm{m}^{(t)})^\mathsf{H}\bm{h}^{(t)}_{\mathrm{UL},k}\right|^2$.

Similarly, to perfectly compensate the channel fading of downlink transmission, the receive scalar at device $k$ can be designed as 
\begin{align}\label{eq: receiver design}
    b^{(t)}_{\mathrm{DL}, k} = \frac{((\bm{h}^{(t)}_{\mathrm{DL}, k})^\mathsf{H} \bm{u}^{(t)})^\mathsf{H}}{\left|(\bm{h}^{(t)}_{\mathrm{DL}, k})^\mathsf{H} \bm{u}^{(t)}\right|^2},\quad \forall k \in [K].
\end{align}
According to \eqref{eq: estimate G}, the estimate $\hat{G}_{i,k}\left(\hat{s}^{(t)}(i)\right)$ at device $k$ can be rewritten as
\begin{align}\label{eq: estimate G1}
    \hat{G}_{i,k}\left(\hat{s}^{(t)}(i)\right) = G_i\left(\hat{s}^{(t)}(i)\right) + n^{(t)}_{\mathrm{DL},k}(i) = G_i\left(s^{(t)}(i) + n^{(t)}_\mathrm{UL}(i)\right) + n^{(t)}_{\mathrm{DL},k}(i).
\end{align}
Different from \eqref{eq: estimate s}, the uplink noise is embedded in function $G(\cdot)$.
In order to characterize the impact of the effective noise of uplink and downlink, we have the following lemma:
\begin{lemma}\label{lemma}
	Suppose that the amplitude of uplink noise is small enough, the estimate $\hat{G}_i\left(\hat{s}^{(t)}(i)\right)$ becomes an unbiased estimation of $G_i\left(s^{(t)}(i)\right)$.
\end{lemma}
To verify Lemma 1, we can write the estimate $\hat{G}_{i,k}\left(\hat{s}^{(t)}(i)\right)$ as its first-order Taylor expansion as follows
\begin{align}\label{eq: taylor}
	\hat{G}_{i,k}\left(\hat{s}^{(t)}(i)\right) = & G_i(s^{(t)}(i) + n^{(t)}_{\mathrm{UL}}) + n^{(t)}_{\mathrm{DL},k}(i) \nonumber \\
	=& G_i\left(s^{(t)}(i)\right) + G'_i(s^{(t)}(i))n^{(t)}_{\mathrm{UL}}(i) + \mathcal{O}(|n^{(t)}_{\mathrm{UL}}(i)|^2) + n^{(t)}_{\mathrm{DL},k}(i) \nonumber \\
	\approx& G_i(s^{(t)}(i)) + G'_i(s^{(t)}(i))n^{(t)}_{\mathrm{UL}}(i) + n^{(t)}_{\mathrm{DL},k}(i),
\end{align}
where $G_i'(\cdot)$ is the first derivative of $G_i(\cdot)$.
Assuming that the amplitude of uplink noise is small, the term $\mathcal{O}(|n^{(t)}_{\mathrm{UL}}(i)|^2)$ is neglected, which implies the last approximated equality in \eqref{eq: taylor} \cite{DBLP:journals/neco/Bishop95}.
Actually, we can perform beamforming design and capacity allocation to control the amplitude of the effective noise of uplink and downlink, which is introduced in Section IV.

%The assumption is to simplify subsequent convergence analysis.
%Since both $n^{(t)}_{\mathrm{UL}}(i)$ and $n^{(t)}_{\mathrm{UL}}(i)$ are zero-mean,

\begin{remark}
Through zero-forcing precoding, we can perfectly compensate the channel fading and thus minimize the impact caused by the channel distortion, which gives the closed-form solution with respect to $\eta^{(t)}$.
As a sequence, the communication noises $n^{(t)}_\mathrm{UL}(i)$ and $n^{(t)}_\mathrm{DL}(i)$ can equivalently represent the impact caused by uplink and downlink communications, i.e., \eqref{eq: estimate s} and \eqref{eq: estimate G1}, respectively. 
As discussed later in this section, the communication noises significantly degrade the convergence performance of the vertical FL algorithm in terms of the optimality gap over the AirComp-based Cloud-RAN.
It is noteworthy that a similar type of aggregation noise was also considered in \cite{DBLP:journals/tcom/AngCZCWY20}. 
However, they addressed the aggregation noise from the algorithmic perspective, while this paper aims to reduce the noise impact through the design of the communication system. 
In addition, the effect of aggregation noise on vertical FL algorithms is different from that on horizontal FL algorithms.
\end{remark}

\subsection{Convergence Result}
In the following, we first provide the convergence analysis of Algorithm \ref{algorithm 1} under AirComp based Cloud-RAN with a non-zero optimality gap from the optimal solution.
We first present some assumptions which are widely adopted in FL literature \cite{DBLP:conf/iclr/LiHYWZ20, DBLP:journals/twc/ChenYSYPC21}.
% \textcolor{red}{Comment: I guess you need firstly assume the differentiability of function $F$ such that the following two assumptions are made sense?}
\begin{assumption}[$\alpha$-strongly convexity]\label{ass:1}
The function $F(\cdot)$ is assumed to be $\alpha$-strongly convex on $\R^d$ with constant $\alpha$, namely, satisfying the following inequality
\begin{align}\label{eq: convex}
F(\bm{y}) \geq F(\bm{x}) + \nabla F(\bm{x})^{\mathsf{T}} (\bm{y} - \bm{x}) + \frac{\alpha}{2} \Vert \bm{y} - \bm{x} \Vert^2,
\end{align}
for all $\bm{x}, \bm{y} \in  \R^d$.
\end{assumption}
\begin{assumption}[$\beta$-smoothness]\label{ass:2}
The function $F(\cdot)$ is assumed to be $\beta$-smooth on $\R^d$ with constant $\beta$, namely, satisfying the following inequality
\begin{align}
F(\bm{y}) \leq F(\bm{x}) + \nabla F(\bm{x})^{\mathsf{T}} (\bm{y} - \bm{x}) + \frac{\beta}{2} \Vert \bm{y} - \bm{x} \Vert^2,
\end{align}
for all $\bm{x}, \bm{y} \in  \R^d$.
\end{assumption}

Based on Assumptions \ref{ass:1}-\ref{ass:2}, we analyze the performance of Algorithm \ref{algorithm 1} under the AirComp based Cloud-RAN communication system proposed in the previous section.
Specifically, the optimality gap from the optimal objective value $F(\bm{w}^*)$  with respect to the system parameters $ \{\bm{m}^{(t)}, \bm{u}^{(t)}, \bm{Q}^{(t)}_{\mathrm{UL}},  \bm{Q}^{(t)}_{\mathrm{DL}} \}$ is established.
The details are summarized in the following theorem.

\begin{theorem}[Convergence of Algorithm \ref{algorithm 1} under AirComp based Cloud-RAN]\label{the1}
Consider Assumptions \ref{ass:1}-\ref{ass:2} with constant learning rate $\mu^{(t)} = 1/\beta$, and also assuming that estimate $\hat{G}_i\left(\hat{s}^{(t)}(i)\right)$ is an unbiased estimation of $G_i\left(s^{(t)}(i)\right)$ according to Lemma \ref{lemma}.
The expected optimality gap for Algorithm \ref{algorithm 1} under AirComp based Cloud-RAN is upper bounded as
\begin{align}
\mathbb{E}\left[F(\bm{w}^{(T)}) - F(\bm{w}^*)\right] &\leq \rho^T\mathbb{E}\left[F(\bm{w}^{(0)}) - F(\bm{w}^*)\right] + \frac{3}{2L^2\beta} B(T),
\end{align}
where $\rho = 1 - \frac{\alpha}{\beta}$ is the contraction rate and the optimality gap $B(T)$ in the $T$-th communication round is given by
\begin{align}\label{eq: optimality gap}
	B(T) = \sum_{t=0}^{T-1}\rho^{T-t-1}\sum_{k=1}^K(\Phi_{1,k}\sigma^2_{\mathrm{UL}}(t) + \Phi_{2,k}\sigma^2_{\mathrm{DL},k}(t)),
\end{align}
where 
\begin{subequations}
\begin{align}
	\Phi_{1,k}(t) = \sum_{i=1}^LG'_i\left(s^{(t)}(i)\right)^2\left\|\bm{x}_{i,k}\right\|^2,	\Phi_{2,k}(t) = \sum_{i=1}^L \left\|\bm{x}_{i,k}\right\|^2. \label{eq: phi}
\end{align}
\end{subequations}
\end{theorem}
\begin{proof}
	Please refer to Appendix \ref{sec: proof of theorem 1}.
\end{proof}

\begin{remark}
Comparing Theorem \ref{the1} to the convergence results of the vanilla GD algorithm \cite{DBLP:journals/siamsc/FriedlanderS12} with error-free communication, we observe that Algorithm \ref{algorithm 1} under AirComp based Cloud-RAN achieves the same convergence rate, but with a non-zero optimality gap $B(T)$ depending on the communication noise variances $\{\sigma^2_{\mathrm{UL}}(t)\}^T_{t=0}$ and $\{\sigma^2_{\mathrm{DL},t}\}^T_{t=0}$.
Hence, the non-zero gap $B(T)$ reveals the impact of the communication noises on the convergence performance.
In particular, a smaller $B(T)$ leads to a smaller gap between $F(\bm{w}^{(T)})$ and $F(\bm{w}^*)$, which motivates us to treat the optimality gap $B(T)$ as the metric of vertical FL performance over the Cloud-RAN.
Based on \eqref{eq: uplink noise} and \eqref{eq: downlink noise}, we can control the noise variances by jointly designing the Cloud-RAN system parameters $\{\bm{m}^{(t)}, \bm{u}^{(t)}, \bm{Q}^{(t)}_{\mathrm{UL}}, \bm{Q}^{(t)}_{\mathrm{DL}}\}$.
\end{remark}

\section{System Optimization}\label{sec: system optimization}
This section presents the system optimization for vertical FL over the AirComp-based Cloud-RAN.
We first formulate the resource allocation problem by minimizing the optimality gap under the system constraints including the limited fronthaul capacity and the transmit power constraints.
Then, we decompose the resulting optimization problem into two sub-problems, i.e., uplink and downlink optimization problems which are solved by the successive convex approximation (SCA) and alternate convex search (ACS) approaches.

\subsection{Problem Formulation}
In this section, we optimize the Cloud-RAN assisted vertical FL system by minimizing the global loss $F(\bm{w})$ based on the convergence results in the previous section.
Specifically, we minimize the optimality gap in terms of $B(T)$ in Theorem \ref{the1} while satisfying the communication constraints (e.g., \eqref{eq: uplink power constraint}, \eqref{eq: uplink capacity constraint}, \eqref{eq: downlink power constraint}, \eqref{eq: downlink capacity constraint}).
Then, the corresponding optimization problem can be formulated as
\begin{equation}\label{prob: original problem1}
\begin{array}{cl}
\underset{\{\bm{m}^{(t)}, \bm{u}^{(t)}, \bm{Q}^{(t)}_{\mathrm{UL}},  \bm{Q}^{(t)}_{\mathrm{DL}} \}}{\operatorname{min}} &
\sum_{t=0}^{T-1}\rho^{T-t-1}\sum_{k=1}^K(\Phi_{1,k}\sigma^2_{\mathrm{UL}}(t) + \Phi_{2,k}\sigma^2_{\mathrm{DL},k}(t)) \\
\mathrm{s.t.} & \sum_{n=1}^N C^{(t)}_{\mathrm{UL}, n} \leq C, \sum_{n=1}^N C^{(t)}_{\mathrm{DL}, n} \leq C, \forall t\in[T] \\
& \bm{Q}^{(t)}_{\mathrm{UL}}(i,i) > 0, \bm{Q}^{(t)}_{\mathrm{UL}}(i,j) = 0, \forall i \neq j, \forall t\in[T]\\
& \bm{Q}^{(t)}_{\mathrm{DL}}(i,i) > 0, \bm{Q}^{(t)}_{\mathrm{DL}}(i,j) = 0, \forall i \neq j, \forall t\in[T]\\
& \|\bm{u}^{(t)}\|^2+ \trace(\bm{Q}^{(t)}_{\mathrm{DL}}) \leq P_{\mathrm{DL}}, \forall t\in[T].
\end{array}
\end{equation}
% where $\Pi^{(t)}$ denotes the set of optimization variables $\{\bm{m}^{(t)}, \bm{u}^{(t)}, \bm{Q}^{(t)}_{\mathrm{UL}},  \bm{Q}^{(t)}_{\mathrm{DL}} \}$.
% Recalling Eq. \eqref{eq: optimality gap}, the objective function is expressed as the weighted summation of $T$ independent non-negative terms.
Since the communication constraints are independent during $T$ rounds, we can decompose the above problem into $T$ independent sub-problems.
% Based on the CSI at the $t$-th communication round, we can solve the $t$-th sub-problem, which is formulated as follows:
% \begin{equation}
% \begin{array}{cl}\label{prob: original problem1}
% \underset{\Pi^{(t)}}{\operatorname{min}}& \sum_{k=1}^K(\Phi_{1,k}(t)\sigma^2_{\mathrm{UL}}(t) + \Phi_{2,k}(t)\sigma^2_{\mathrm{DL},k}(t))\\
% \mathrm{s.t.}
% & \sum_{n=1}^N C^{(t)}_{\mathrm{UL}, n} \leq C, \sum_{n=1}^N C^{(t)}_{\mathrm{DL}, n} \leq C, \\
% & \bm{Q}^{(t)}_{\mathrm{UL}}(i,i) > 0, \bm{Q}^{(t)}_{\mathrm{UL}}(i,j) = 0, \forall i \neq j,\\
% & \bm{Q}^{(t)}_{\mathrm{DL}}(i,i) > 0, \bm{Q}^{(t)}_{\mathrm{DL}}(i,j) = 0, \forall i \neq j,\\
% & \|\bm{u}^{(t)}\|^2+ \trace(\bm{Q}^{(t)}_{\mathrm{DL}}) \leq P_{\mathrm{DL}},
% \end{array}
% \end{equation}
% where $\Phi_1(t)$ and $\Phi_2(t)$ are given in \eqref{eq: phi}.
Additionally, note that both the objective function and the constraints can be divided into uplink and downlink transmissions, which are independent of each other.
As a consequence, we solve problem \eqref{prob: original problem1} by solving the following two sub-problems in parallel at each communication round,
\begin{equation}\label{prob: uplink optimization}
\begin{array}{cl}
\underset{\bm{m}^{(t)}, \bm{Q}^{(t)}_{\mathrm{UL}} }{\operatorname{min}}
& \sum_{k=1}^K\Phi_{1,k}(t)\sigma^2_{\mathrm{UL}}(t)\\
\mathrm{s.t.}
& \sum_{n=1}^N C^{(t)}_{\mathrm{UL}, n} \leq C, \\
& \bm{Q}^{(t)}_{\mathrm{UL}}(i,i) > 0, \bm{Q}^{(t)}_{\mathrm{UL}}(i,j) = 0, \forall i \neq j,
\end{array}
\end{equation}
and 
\begin{equation}\label{prob: downlink optimization}
\begin{array}{cl}
\underset{ \bm{u}^{(t)}, \bm{Q}^{(t)}_{\mathrm{DL}}}{\operatorname{min}}& \sum_{k=1}^K \Phi_{2,k}(t)\sigma^2_{\mathrm{DL},k}(t) \\
\mathrm{s.t.}
& \sum_{n=1}^N C^{(t)}_{\mathrm{DL}, n}  \leq C, \\
& \bm{Q}^{(t)}_{\mathrm{DL}}(i,i) > 0, \bm{Q}^{(t)}_{\mathrm{DL}}(i,j) = 0, \forall i \neq j,\\
& \|\bm{u}^{(t)}\|^2+ \trace(\bm{Q}^{(t)}_{\mathrm{DL}}) \leq P_{\mathrm{DL}}.
\end{array}
\end{equation}

\subsection{Optimization Framework}
In the following, we specify the optimization framework for solving the uplink and downlink optimization problems, respectively.
To simplify the notations, we omit the time index of the optimized variables in the following as we focus on the system optimization problem at one communication round.
\subsubsection{Uplink Optimization}
By substituting \eqref{eq: uplink capacity constraint} and \eqref{eq: uplink noise} into problem \eqref{prob: uplink optimization}, we then obtain the following optimization problem in the uplink transmission:
\begin{equation}\label{prob: uplink optimization1}
	\begin{aligned}
\min_{ \bm{m}, \bm{Q}_{\mathrm{UL}}}&\quad \frac{1}{2\eta}\bm{m}^\mathsf{H}\left(\sigma^2_z\bm{I}+\bm{Q}_\mathrm{UL}\right)\bm{m} \\
\text{s.t.}\;\;
&\log\frac{\left|P_{\mathrm{UL}}\sum_{k=1}^{K} \bm{h}_{\mathrm{UL}, k} (\bm{h}_{\mathrm{UL}, k})^\mathsf{H}+\sigma^2_z\bm{I}+\bm{Q}_{\mathrm{UL}}\right|}{\left|\bm{Q}_{\mathrm{UL}}\right|} \leq C, \\
&\bm{Q}_{\mathrm{UL}}(i,i) > 0, \bm{Q}_{\mathrm{UL}}(i,j) = 0,\quad \forall i \neq j,
\end{aligned}
\end{equation}
where $\eta = P_{\mathrm{UL}} \min_k |(\bm{m})^\mathsf{H}\bm{h}_{\mathrm{UL},k}|^2$.

We propose to adopt the alternating optimization technique to solve the uplink optimization problem \eqref{prob: uplink optimization1}.
Specifically, we first fix the covariance matrix $\bm{Q}_{\mathrm{UL}}$ in problem \eqref{prob: uplink optimization1} and optimize the receive beamforming vector $\bm{m}$ by solving the following min-max problem:
\begin{equation}\label{prob: uplink optimization2}
\begin{aligned}
\qquad\min_{ \bm{m}} \max_{k} \quad \frac{\bm{m}^\mathsf{H}\tilde{\bm{Q}}\bm{m}}{{|\bm{m}^\mathsf{H}\bm{h}_{\mathrm{ul,k}}|^2}},\quad\text{s.t.}\quad\tilde{\bm{Q}} = \sigma^2_z\bm{I}+\bm{Q}_\mathrm{UL},
\end{aligned}
\end{equation}
%where $\tilde{\bm{Q}} = \sigma^2_z\bm{I}+\bm{Q}_\mathrm{UL}$.
% \begin{proposition}\label{prop1}
% 	Problem \eqref{prob: uplink optimization2} is equivalent to
% \begin{equation}\label{prob: uplink optimization3}
% \min_{\bm{m}} \max_{k} \quad -|\bm{m}^\mathsf{H}\bm{h}_{\mathrm{ul,k}}|^2, \quad
% \mathrm{s.t.} \quad \bm{m}^\mathsf{H}\tilde{\bm{Q}}\bm{m} = 1.
% \end{equation}
% \end{proposition}
% \begin{proof}
% 	% Please refer to Appendix \ref{proof: prop1}.
% 	Please refer to \cite{DBLP:journals/tcom/FangJSZCL21}
% \end{proof}
% According to \cite{DBLP:journals/tcom/FangJSZCL21}, Problem \eqref{prob: uplink optimization2} is equivalent to
which is equivalent to
\begin{equation}\label{prob: uplink optimization3}
\min_{\bm{m}} \max_{k} \quad -|\bm{m}^\mathsf{H}\bm{h}_{\mathrm{ul,k}}|^2, \quad
\mathrm{s.t.} \quad \bm{m}^\mathsf{H}\tilde{\bm{Q}}\bm{m} = 1.
\end{equation}

Problem \eqref{prob: uplink optimization3} is a quadratically constrained quadratic programming (QCQP) problem with a non-convex constraint. 
To tackle the non-convex constraint in problem \eqref{prob: uplink optimization3}, we relax the equality constraint to an inequality constraint, i.e., $\bm{m}^\mathsf{H}\tilde{\bm{Q}}\bm{m} \leq 1$, which yields the following relaxed problem:
\begin{equation}\label{prob: uplink optimization4}
\min_{\bm{m}} \max_{k} \quad -|\bm{m}^\mathsf{H}\bm{h}_{\mathrm{ul,k}}|^2, \quad
\mathrm{s.t.} \quad \bm{m}^\mathsf{H}\tilde{\bm{Q}}\bm{m} \leq 1.
\end{equation}
% Although problem \eqref{prob: uplink optimization4} can be solved by the semi-definite relaxation (SDR) technique \cite{DBLP:journals/spm/LuoMSYZ10} by dropping the rank-one constraint in the lifted matrix optimization problem, its performance deteriorates quickly as the dimension of optimization variable increases, e.g., $M$ and $N$ \cite{DBLP:journals/siamjo/LuoSTZ07}.
% We adopt the successive convex approximation (SCA) technique \cite{DBLP:journals/spl/MehannaHGKS15} to solve problem \eqref{prob: uplink optimization4} with good performance and low computational complexity.
We adopt the successive convex approximation (SCA) technique to solve problem \eqref{prob: uplink optimization4} with good performance and low computational complexity.

To develop the SCA approach, we convert problem \eqref{prob: uplink optimization4} from the complex domain to the real domain with the following variables:
\begin{subequations}
\begin{align}
	\tilde{\bm{m}} &= \left[\Re(\bm{m})^\mathsf{T} \Im(\bm{m})^\mathsf{T} \right]^\mathsf{T}, \\
	\tilde{\bm{H}}_k &= -\begin{bmatrix}
		\Re(\bm{h}_{\mathrm{UL},k}\bm{h}^\mathsf{H}_{\mathrm{UL},k}) &-\Im(\bm{h}_{\mathrm{UL},k}\bm{h}^\mathsf{H}_{\mathrm{UL},k}) \\
		\Im(\bm{h}_{\mathrm{UL},k}\bm{h}^\mathsf{H}_{\mathrm{UL},k}) &\Re(\bm{h}_{\mathrm{UL},k}\bm{h}^\mathsf{H}_{\mathrm{UL},k})
	\end{bmatrix}, \forall k \in [K], \\
	\tilde{\bm{Q}'} &= \begin{bmatrix}
		\Re(\tilde{\bm{Q}}) &-\Im(\tilde{\bm{Q}}) \\
		\Im(\tilde{\bm{Q}}) &\Re(\tilde{\bm{Q}})
	\end{bmatrix}.
\end{align}
\end{subequations}
Then, problem \eqref{prob: uplink optimization4} can be reformulated as follows:
\begin{equation}\label{prob: uplink optimization5}
\min_{\tilde{\bm{m}}} \max_{k} \quad \tilde{\bm{m}}^\mathsf{T}\tilde{\bm{H}}_k\tilde{\bm{m}}, \quad
\mathrm{s.t.} \quad \tilde{\bm{m}}^\mathsf{T}\tilde{\bm{Q}'}\tilde{\bm{m}} \leq 1.
\end{equation}

For simplicity, we define $u_k(\tilde{\bm{m}}) = \tilde{\bm{m}}^\mathsf{T}\tilde{\bm{H}}_k\tilde{\bm{m}}$.
Since $u_k(\tilde{\bm{m}})$ is a concave function, we can upper bound $u_k(\tilde{\bm{m}})$ by the tangent of the $l$-th point $\bm{m}_{(l)}$ as follows:
\begin{align*}
	u_k(\tilde{\bm{m}}) &\leq \nabla u_k(\tilde{\bm{m}}_{(l)})^\mathsf{T}(\tilde{\bm{m}} - \tilde{\bm{m}}_{(l)}) + u_k(\tilde{\bm{m}}_{(l)}) \\
	&= \left(2\tilde{\bm{H}}_k\tilde{\bm{m}}_{(l)}\right)^\mathsf{T}\tilde{\bm{m}} - (\tilde{\bm{m}}_{(l)})^\mathsf{T}\tilde{\bm{H}}_k\tilde{\bm{m}}_{(l)}.
\end{align*}
The sub-problem for each SCA iteration can be reformulated as follows:
\begin{equation}\label{prob: uplink optimization6}
\begin{aligned}
\min_{\tilde{\bm{m}}} \max_{k}
	&\quad  \left(2\tilde{\bm{H}}_k\tilde{\bm{m}}_{(l)}\right)^\mathsf{T}\tilde{\bm{m}} - (\tilde{\bm{m}}_{(l)})^\mathsf{T}\tilde{\bm{H}}_k\tilde{\bm{m}}_{(l)} \\
\mathrm{s.t.}
	& \quad \tilde{\bm{m}}^\mathsf{T}\tilde{\bm{Q}'}\tilde{\bm{m}} \leq 1.
\end{aligned}
\end{equation}
Starting from an initial point $\tilde{\bm{m}}_{(0)}$, a sequence of solutions $\{\tilde{\bm{m}}_{(l)}\}$ is generated by solving a series of the above sub-problems.

When fixing the receive beamforming vector $\bm{m}$, problem \eqref{prob: uplink optimization1} is reduced to the following problem:
\begin{equation}\label{prob: uplink optimization7}
\begin{aligned}
\min_{\bm{Q}_{\mathrm{UL}}}& \quad \bm{m}^\mathsf{H}\left(\sigma^2_z\bm{I}+\bm{Q}_\mathrm{UL}\right)\bm{m} \\
\mathrm{s.t.}
& \quad \log\frac{\left|P_{\mathrm{UL}}\sum_{k=1}^{K} \bm{h}_{\mathrm{UL}, k} (\bm{h}_{\mathrm{UL}, k})^\mathsf{H}+\sigma^2_z\bm{I}+\bm{Q}_{\mathrm{UL}}\right|}{\left|\bm{Q}_{\mathrm{UL}}\right|} \leq C, \\
& \quad \bm{Q}_{\mathrm{UL}}(i,i) > 0,  \bm{Q}_{\mathrm{UL}}(i,j) = 0,\quad \forall i \neq j.
\end{aligned}
\end{equation}
It is easy to verify that problem \eqref{prob: uplink optimization7} is convex with respect to $\bm{Q}_\mathrm{UL}$, which can be solved efficiently with polynomial complexity by using CVX, a package for specifying and solving convex programs \cite{cvx}.

The proposed algorithm for uplink optimization is summarized in Algorithm \ref{algorithm 2}.
The convergence results are summarized as follows:
\begin{itemize}
	\item In the inner loop (Steps 4-7), the values of objective function in problem \eqref{prob: uplink optimization5} achieved by sequence $\{\tilde{\bm{m}}^{(i)}_{(l)}\}_{l=0}^\infty$ converge monotonically.
	\item In the outer loop (Steps 2-11), the values of objective function in problem \eqref{prob: uplink optimization3} achieved by sequence $\{\tilde{\bm{m}}^{(i)}_{(0)}, \bm{Q}_\mathrm{UL}^{(i)}\}_{i=0}^\infty$ converge monotonically.
\end{itemize}
% The proof of the convergence can be referred to \cite{DBLP:journals/tcom/FangJSZCL21}.
Problem \eqref{prob: uplink optimization6} can be equivalently formulated as a convex QCQP problem and then be solved by using the interior-point method, and the computational complexity is $\mathcal{O}((MN+K)^{3.5})$.
Similarly, the complexity of solving problem \eqref{prob: uplink optimization7} is $\mathcal{O}((MN)^{3.5})$.

\begin{algorithm}[tbp]
\setstretch{0.95}
\caption{Proposed Algorithm for Uplink Optimization \eqref{prob: uplink optimization1}} 
\label{algorithm 2}
%\DontPrintSemicolon
\SetAlgoLined
\SetKwInOut{Input}{Input}
\SetKwInOut{Output}{Output}
\Input{Initial points $\tilde{\bm{m}}^{(0)}_{(0)}, \bm{Q}^{(0)}_{\mathrm{UL}}$ and threshold $\epsilon$\;}
Set $i = 0$\;
\Repeat{$\left\|\tilde{\bm{m}}^{(i+1)}_{(0)} - \tilde{\bm{m}}^{(i)}_{(0)}\right\| + \left\|\bm{Q}_\mathrm{UL}^{(i+1)} - \bm{Q}_\mathrm{UL}^{(i)} \right\|< \epsilon$}{
	Set $l = 0$\;
	\Repeat{$\left\|\tilde{\bm{m}}^{(i)}_{(l+1)} - \tilde{\bm{m}}^{(i)}_{(l)}\right\| < \epsilon$}{
		Update $\tilde{\bm{m}}^{(i)}_{(l+1)}$ by solving problem \eqref{prob: uplink optimization6}\;
		$l \leftarrow l + 1$\;
    }
	Set $\tilde{\bm{m}}^{(i+1)}_{(0)} \leftarrow \tilde{\bm{m}}^{(i)}_{(l)}$\;
	Update $\bm{Q}_{\mathrm{UL}}^{(i+1)}$ by solving problem \eqref{prob: uplink optimization7}\;
	Set $i \leftarrow i + 1$\;
}
\end{algorithm}

\subsubsection{Downlink Optimization}
By substituting \eqref{eq: downlink noise} into problem \eqref{prob: downlink optimization}, we solve the following problem:
\begin{equation}\label{prob: downlink optimization1}
	\begin{aligned}
\min_{ \bm{u}, \bm{Q}_{\mathrm{DL}}  }& \quad \sum_{k=1}^K  \frac{\Phi_{2,k}}{|(\bm{h}^{(t)}_{\mathrm{DL}, k})^\mathsf{H} \bm{u}|^2}\left(\sigma_{z}^2+(\bm{h}_{\mathrm{DL}, k})^\mathsf{H} \bm{Q}_\mathrm{DL} \bm{h}_{\mathrm{DL}, k}\right) \\
\text{s.t.}\; 
& \quad \log \frac{\left|\bm{u}\bm{u}^\mathsf{H} + \bm{Q}_{\mathrm{DL}}\right|}{\left|\bm{Q}_{\mathrm{DL}}\right|} \leq C, \\
&\quad \bm{Q}_{\mathrm{DL}}(i,i) > 0, \bm{Q}_{\mathrm{DL}}(i,j) = 0, \forall i \neq j,\\
&\quad \|\bm{u}\|^2+ \trace(\bm{Q}_{\mathrm{DL}}) \leq P_{\mathrm{DL}}.
\end{aligned}
\end{equation}

Although the objective function in problem \eqref{prob: downlink optimization1} is convex with respect to each optimization variable when other optimization variables are fixed, the capacity constraint is non-convex with respect to $\bm{u}^{(t)}$.
Hence, we approximate  problem \eqref{prob: downlink optimization1} to a convex optimization by linearizing the limited fronthaul capacity constraint.
By solving the approximated optimization problem, we successively approximate the optimal solution of the original problem.

Before reformulating the downlink optimization problem, we first present the following lemma for convex approximation, which is also adopted in \cite{DBLP:journals/jsac/ZhouY14}.
\begin{lemma}
	For positive definite Hermitian matrices $\bm{\Omega}, \bm{\Sigma} \in \mathbb{C}^{N\times N}$, we have
	\begin{align}
		\log |\bm{\Omega}| \leq \log |\bm{\Sigma}| + \trace\left(\bm{\Sigma}^{-1}\bm{\Omega}\right) - N,
	\end{align}
	with equality if and only if $\bm{\Omega} = \bm{\Sigma}$.
\end{lemma}

By applying Lemma 1 to the capacity constraint in problem \eqref{prob: downlink optimization1} and setting $\bm{\Omega} = \bm{u}\bm{u}^\mathsf{H} + \bm{Q}_{\mathrm{DL}}$, we can approximate the capacity constraint in problem \eqref{prob: downlink optimization1} with the following convex constraint:
\begin{align}\label{eq: convex constraint}
	\log|\bm{\Sigma}| + \trace\left(\bm{\Sigma}^{-1}\left(\bm{u}\bm{u}^\mathsf{H} + \bm{Q}_{\mathrm{DL}}\right)\right)  - \log|\bm{Q}_\mathrm{DL}|\leq C + MN.
\end{align}
Note that the capacity constraint in problem \eqref{prob: downlink optimization1} is always feasible when the convex constraint \eqref{eq: convex constraint} is feasible.
The two constraints are equivalent when 
\begin{align}\label{eq: optimal sigma}
	\bm{\Sigma}^* = \bm{u}\bm{u}^\mathsf{H} + \bm{Q}_{\mathrm{DL}}.
\end{align}

By using the convex constraint \eqref{eq: convex constraint} to replace the capacity constraint in \eqref{prob: downlink optimization1}, we reformulate the downlink optimization problem to the following equivalent problem:
\begin{equation}\label{prob: downlink optimization2}
	\begin{aligned}
\min_{\bm{u}, \bm{Q}_{\mathrm{DL}}, \bm{\Sigma}}& \quad \trace \left(\bm{Q}_\mathrm{DL}\sum_{k=1}^K  \frac{\Phi_{2,k}}{\bm{u}^\mathsf{H}\bm{H}_{\mathrm{DL},k}\bm{u}}\bm{H}_{\mathrm{DL},k}\right) \\
\text{s.t.}\;\;
&\quad \log|\bm{\Sigma}| + \trace\left(\bm{\Sigma}^{-1}\bm{\Omega}\right)- \log|\bm{Q}_\mathrm{DL}|\leq C + MN, \\
&\quad \bm{Q}_{\mathrm{DL}}(i,i) > 0,  \bm{Q}_{\mathrm{DL}}(i,j) = 0, \forall i \neq j,\\
&\quad \|\bm{u}\|^2+ \trace(\bm{Q}_{\mathrm{DL}}) \leq P_{\mathrm{DL}},
\end{aligned}
\end{equation}
where $\bm{\Omega} = \bm{u}\bm{u}^\mathsf{H} + \bm{Q}_{\mathrm{DL}}$ and $\bm{H}_{\mathrm{DL},k} = \bm{h}_{\mathrm{DL},k} \bm{h}_{\mathrm{DL},k}^\mathsf{H}$.

In this paper, we propose to alternatively optimize problem \eqref{prob: downlink optimization2}.
When $\{\bm{u}, \bm{Q}_{\mathrm{DL}}\}$ are fixed, the optimal value of $\bm{\Sigma}$ is given by \eqref{eq: optimal sigma}.
When $\bm{\Sigma}$ is fixed, the solutions of $\{\bm{u}, \bm{Q}_{\mathrm{DL}}\}$ are given by solving the following problem:
\begin{equation}\label{prob: downlink optimization3}
	\begin{aligned}
\min_{\bm{u}, \bm{Q}_{\mathrm{DL}}}& \quad \trace \left(\bm{Q}_\mathrm{DL}\sum_{k=1}^K  \frac{\Phi_{2,k}}{\bm{u}^\mathsf{H}\bm{H}_{\mathrm{DL},k}\bm{u}}\bm{H}_{\mathrm{DL},k}\right) \\
\text{s.t.}\; 
&\quad \trace\left(\bm{\Sigma}^{-1}\bm{\Omega}\right)- \log|\bm{Q}_\mathrm{DL}|\leq C + MN - \log|\bm{\Sigma}|, \\
&\quad \bm{Q}_{\mathrm{DL}}(i,i) > 0, \bm{Q}_{\mathrm{DL}}(i,j) = 0, \forall i \neq j,\\
&\quad \|\bm{u}\|^2+ \trace(\bm{Q}_{\mathrm{DL}}) \leq P_{\mathrm{DL}}.
\end{aligned}
\end{equation}

The objective function in problem \eqref{prob: downlink optimization3} is a bi-convex problem for $\{\bm{u}, \bm{Q}_{\mathrm{DL}}\}$, e.g., by fixing $\bm{u}$ the function is convex for $\bm{Q}_\mathrm{DL}$, and by fixing $\bm{Q}_\mathrm{DL}$ the function is convex for $\bm{u}$.
By exploiting the bi-convex structure of the problem, we adopt the ACS approach to solve problem \eqref{prob: downlink optimization3}. 
Specifically, problem \eqref{prob: downlink optimization3} is transformed into the following two sub-problems:
\begin{equation}\label{prob: downlink optimization4}
	\begin{aligned}
\min_{\bm{Q}_{\mathrm{DL}}}& \quad \trace \left(\bm{Q}_\mathrm{DL}\sum_{k=1}^K  \frac{\Phi_{2,k}}{\bm{u}^\mathsf{H}\bm{H}_{\mathrm{DL},k}\bm{u}}\bm{H}_{\mathrm{DL},k}\right) \\
\text{s.t.}
&\quad \trace\left(\bm{\Sigma}^{-1}\bm{\Omega}\right)- \log|\bm{Q}_\mathrm{DL}|\leq C + MN - \log|\bm{\Sigma}|, \\
&\quad \bm{Q}_{\mathrm{DL}}(i,i) > 0, \bm{Q}_{\mathrm{DL}}(i,j) = 0, \forall i \neq j,\\
&\quad \|\bm{u}\|^2+ \trace(\bm{Q}_{\mathrm{DL}}) \leq P_{\mathrm{DL}},
\end{aligned}
\end{equation}
and 
\begin{equation}\label{prob: downlink optimization5}
	\begin{aligned}
\min_{\bm{u}}& \quad \trace \left(\bm{Q}_\mathrm{DL}\sum_{k=1}^K  \frac{\Phi_{2,k}}{\bm{u}^\mathsf{H}\bm{H}_{\mathrm{DL},k}\bm{u}}\bm{H}_{\mathrm{DL},k}\right) \\
\text{s.t.}
&\quad \|\bm{u}\|^2+ \trace(\bm{Q}_{\mathrm{DL}}) \leq P_{\mathrm{DL}}.
\end{aligned}
\end{equation}
The overall algorithm for downlink optimization is summarized in Algorithm \ref{algorithm 3}.
The sub-problems \eqref{prob: downlink optimization4} and \eqref{prob: downlink optimization5} are convex and can be solved by using the interior-point method within $\mathcal{O}((MN)^{3.5})$ iterations.
The convergence of ACS has already been intensively studied in \cite{DBLP:journals/mmor/GorskiPK07}.

% The convergence of ACS has already been intensively studied in \cite{DBLP:journals/mmor/GorskiPK07}.
% Since the objective function has a general lower bound 0, based on \cite[Theorem 4.5]{DBLP:journals/mmor/GorskiPK07}, the optimization procedure of the proposed algorithm is guaranteed to converge monotonically in terms of the objective function.
% Furthermore, if we have the unique optimal solution of $\{\bm{u}, \bm{Q}_\mathrm{DL}\}$, according to \cite[Theorem 4.9]{DBLP:journals/mmor/GorskiPK07}, the sequence $\{\bm{u}^{(i)}, \bm{Q}^{(i)}_{\mathrm{DL}}\}$ generated by the loop (Steps 2-7) of Algorithm \ref{algorithm 3} satisfies:
% \begin{align}
% 	\lim_{i\rightarrow\infty} \|\bm{u}^{(i+1)} - \bm{u}^{(i)}\| + \|\bm{Q}_{\mathrm{DL}}^{(i+1)} - \bm{Q}_{\mathrm{DL}}^{(i)}\| = 0.
% \end{align}

\begin{algorithm}[tbp]
\setstretch{0.95}
\caption{Proposed Algorithm for Downlink Optimization \eqref{prob: downlink optimization1}}
\label{algorithm 3}
\SetAlgoLined
\SetKwInOut{Input}{Input}
\SetKwInOut{Output}{Output}
\Input{Initial points $\bm{\Sigma}^{(0)}$ and threshold $\epsilon$\;}
Set $i = 0$\;
\Repeat{$\left\|\bm{u}^{(i+1)} - \bm{u}^{(i)}\right\| + \left\|\bm{Q}_{\mathrm{DL}}^{(i+1)} - \bm{Q}_{\mathrm{DL}}^{(i)} \right\| + \left\|\bm{\Sigma}^{(i+1)} - \bm{\Sigma}^{(i)}\right\|< \epsilon$}{
	Update $\bm{Q}_{\mathrm{DL}}^{(i+1)}$ by solving problem \eqref{prob: downlink optimization4}\;
	Update $\bm{u}^{(i+1)}$ by solving problem \eqref{prob: downlink optimization5}\;
	Update $\bm{\Sigma}^{(i+1)}$ by using \eqref{eq: optimal sigma}\;
	Set $i \leftarrow i + 1$\;
}
\end{algorithm}

\section{Numerical Results}
In this section, we conduct extensive numerical experiments to evaluate the performance of the proposed system optimization algorithm for the AirComp based Cloud-RAN for wireless vertical FL.

\subsection{Simulation Settings}

\subsubsection{Vertical FL Setting}
We consider a vertical FL setting where $K=49$ devices cooperatively train a regularized logistic regression model.
We simulate the image classification task on Fashion-MNIST dataset.
% \cite{DBLP:journals/corr/abs-1708-07747}.
% We adopt the cross-entropy loss function which has the following form:
% \begin{align*}
% F(\bm{w}) = -\frac{1}{L} \sum_{i=1}^{L} \sum_{j=1}^R y_{ij}\log\left(\frac{\exp\left(\sum_{k=1}^K\bm{w}_{kj}^\mathsf{T}\bm{x}_{i,k}\right)}{\sum_{j'=1}^R\exp\left(\sum_{k=1}^K\bm{w}_{kj'}^\mathsf{T}\bm{x}_{i,k}\right)}\right) + \lambda \sum_{k=1}^{K} \sum_{j=1}^R \| \bm{w}_{kj} \|_2^2,
% \end{align*}
% where $R$ is the number of class.
% Hence, the output dimension of each local prediction function $g_k(\bm{w}, \bm{x}_i)$ is $R$, and we employ element-wise sum to aggregate local prediction results.
The local training data are independent and identically distributed (i.i.d.) drawn from $L=50,000$ images, and the test dataset contains $10,000$ different images. 
For each training image of 784 features, each device is assigned with $784/49=16$ non-overlapped features.
The learning rate $\mu^{(t)}$ is set to 0.01.
We use the optimality gap to characterize the convergence of the proposed algorithm.
Furthermore, the learning performance of wireless vertical FL is evaluated by the training loss for training and the test accuracy for testing. 

\subsubsection{Communication Setting}
We consider a Cloud-RAN assisted vertical FL system with $N$ edge servers and $K=49$ devices randomly located in a circle area of radius $500$ m.
% ($N$ will be specified later).
The edge servers are randomly distributed in the circular area.
Each edge server is equipped with $M$ antennas, and each device is equipped with single antenna.
% ($M$ will be specified later).
% Perfect global CSI is assumed to be available to the central server, and perfect local CSI is assumed to be available at each device.
The channel coefficients are given by the small-scale fading coefficients multiplied by the square root of the path loss, i.e., $30.6 + 36.7\log_{10}(d)$ dB \cite{access2010further}, where $d$ (in meter) is the distance between the device and the edge server. 
The small-scale fading coefficients follow the standard i.i.d. Gaussian distribution.
The transmit power constraint for all device is set to $P_{\mathrm{UL}} = 23$ dBm,
the power spectral density of the background noise at each edge server is assumed to be $-169$ dBm/Hz \cite{DBLP:journals/tsp/LiuZ15}, and the noise figure is $7$ dB.
All numerical results are averaged over $50$ trials.

\subsection{Convergence of the Proposed Algorithm}
\begin{figure}[h]
	\vspace{-0.5cm}
	\centering
	\includegraphics[scale=0.6]{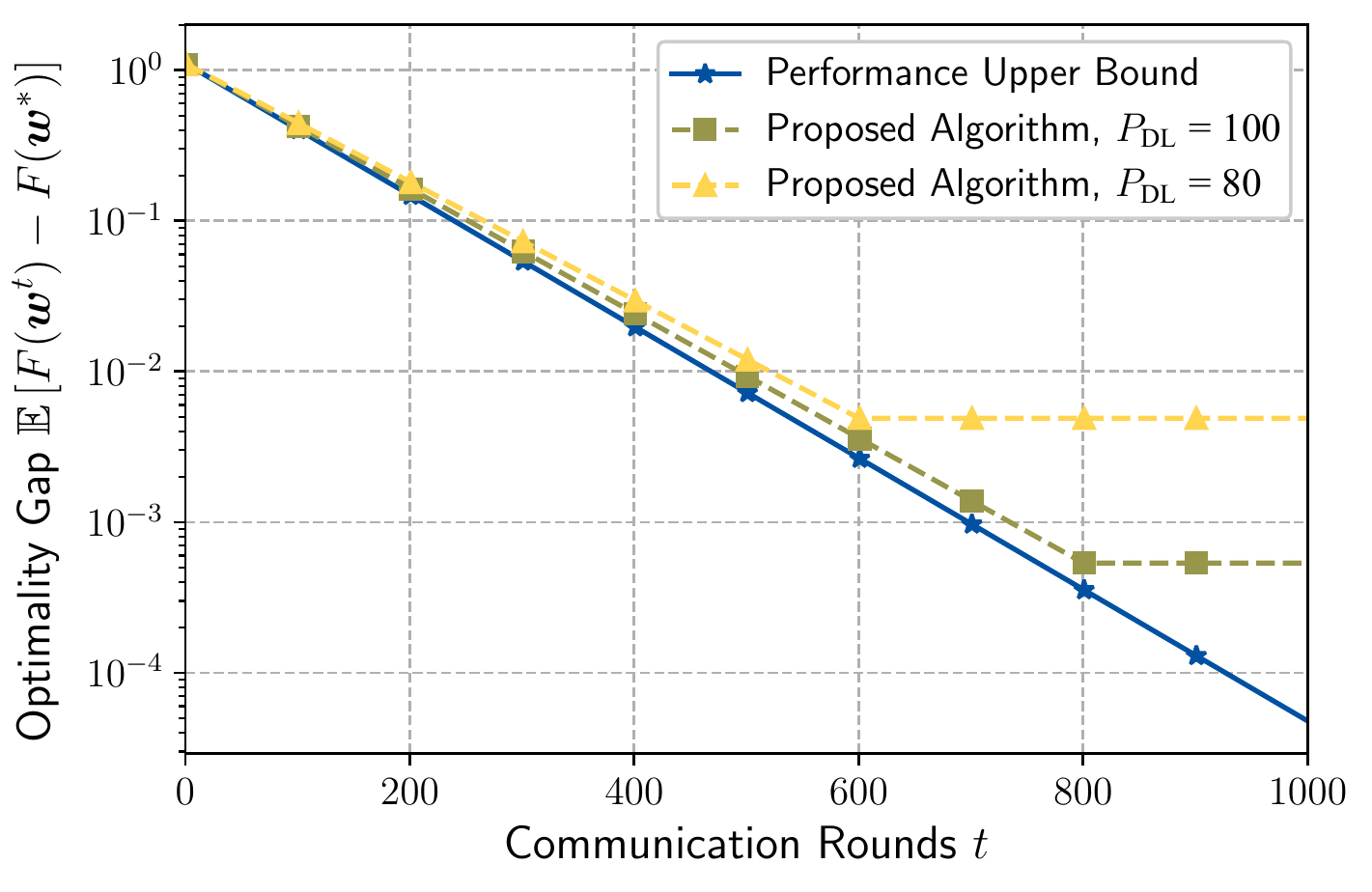}
	\vspace{-0.4cm}
	\caption{Optimality gap versus communication round for different downlink power constraints.}
	\label{fig1}
	\vspace{-0.6cm}
\end{figure}
In this subsection, we examine the convergence performance of the proposed algorithm for image classification task.
We consider Algorithm \ref{algorithm 1} as \textbf{Performance Upper Bound}, where the channels are noiseless (i.e., $\sigma^2_z$ = 0) and the fronthaul capacity is infinite (i.e., $C = +\infty$), which characterizes the best possible learning performance for vertical FL.

In Fig. \ref{fig1}, we plot the optimality gap in $1000$ communication rounds for different downlink power constraints with the capacity constraint $C = 200$ Mbps.
The number of edge servers and the number of antennas at each edge server are set to $N = 8$ and $M = 10$, respectively.
It is observed that the proposed algorithms under two considered downlink power constraints are both able to linearly converge.
This implies that the proposed algorithm can effectively compensate for the distortion of channels and noises, thereby speeding up the convergence and reducing optimality gap.
Furthermore, the case with higher downlink power constraint achieves better performance gains (smaller optimality gap), since the signal-to-noise ratio (SNR) in this case is higher.

\subsection{Impact of Key System Parameters}
In this subsection, we show the performance gain of joint optimization for wireless and fronthaul resource allocation, and investigate the impact of various key system parameters.
For clarity, we refer the proposed algorithm to jointly optimize beamforming vectors and covariance matrices as \textbf{Joint Optimization}.
We also consider different schemes for designing beamforming vectors and covariance matrices as baselines:
\begin{itemize}
	\item \textbf{Baseline 1: Uniform quantization with uniform beamforming.} The transmit beamforming $\bm{u}^{(t)}$ and the receive beamforming $\bm{m}^{(t)}$ are uniformly designed, i.e., $\bm{u}^{(t)} = \kappa\mathbf{1}$ and $\bm{m}^{(t)}=\sqrt{1/MN}\mathbf{1}$, where the scalar $\kappa$ is chosen to satisfy the maximum power constraint \eqref{eq: downlink power constraint}.
	The central server performs uniform quantization noise levels across the antennas for all edge servers, i.e., setting $\bm{Q}_\mathrm{UL}^{(t)} = \lambda_{\mathrm{UL}} \mathbf{I}$ and $\bm{Q}_\mathrm{DL}^{(t)} = \lambda_{\mathrm{DL}} \mathbf{I}$, where the scalars $\lambda_{\mathrm{UL}}$ and $\lambda_{\mathrm{DL}}$ are chosen to satisfy the fronthaul capacity constraints \eqref{eq: uplink capacity constraint} and \eqref{eq: downlink capacity constraint}.
	\item \textbf{Baseline 2: Uniform quantization with optimized beamforming.} This baseline is similar to Baseline 1 except that the receive beamforming vector $\{\bm{m}^{(t)}\}$ and the transmit beamforming vector $\{\bm{u}^{(t)}\}$ are optimized by Algorithms \ref{algorithm 2} and \ref{algorithm 3} with the uniformly designed covariance matrices $\bm{Q}_\mathrm{UL}^{(t)} = \lambda_{\mathrm{UL}} \mathbf{I}$ and $\bm{Q}_\mathrm{DL}^{(t)} = \lambda_{\mathrm{DL}} \mathbf{I}$.
	\item \textbf{Baseline 3: Optimized quantization with uniform beamforming.} This baseline is similar to Baseline 1 except that the covariance matrices $\{\bm{Q}^{(t)}_{\mathrm{UL}}, \bm{Q}^{(t)}_{\mathrm{DL}}\}$ are optimized by Algorithms \ref{algorithm 2} and \ref{algorithm 3} with the uniformly designed beamforming vectors $\bm{u}^{(t)} = \kappa\mathbf{1}$ and $\bm{m}^{(t)}=\sqrt{1/MN}\mathbf{1}$.
\end{itemize}

In order to record the performance of each scheme, each numerical experiment runs $500$ communication rounds, and the performance of the last round is recorded.
Unless specified otherwise, the number of edge servers is set to $8$, i.e., $N=8$, the number of antennas at each edge server is set to $2$, i.e., $M=2$, and the downlink power constraint is set to $P_{\mathrm{DL}}=100$ dBm.

Fig. \ref{fig: capacity} shows the training loss and test accuracy achieved by different schemes versus the fronthaul capacity $C$.
It is observed that the proposed joint optimization for wireless and fronthaul resource allocation achieves higher test accuracy and less training loss over Baselines 1-3, and approaches the performance upper bound when the fronthaul capacity is greater than or equal to 180 Mbps.
Furthermore, Fig. \ref{fig: capacity} shows that in the region of small the fronthaul capacity ($C < 200$ Mbps), Baseline 3 performs closer to the proposed joint optimization algorithm compared with Baselines 1 and 2 using uniform quantization.
We can therefore conclude that the performance gain of joint optimization is mainly obtained from the quantization optimization at edge servers.
On the other hand, as the fronthaul capacity increases, it is observed that Baseline 2 outperforms Baseline 3.
This is due to the impact of quantization error on the performance of vertical FL becomes negligible, and thus the scheme with optimized beamforming can achieve better performance.

\begin{figure}[h]
	\vspace{-0.5cm}
	\centering
	\begin{subfigure}[h]{0.49\textwidth}
	        \centering
	        \includegraphics[scale=0.6]{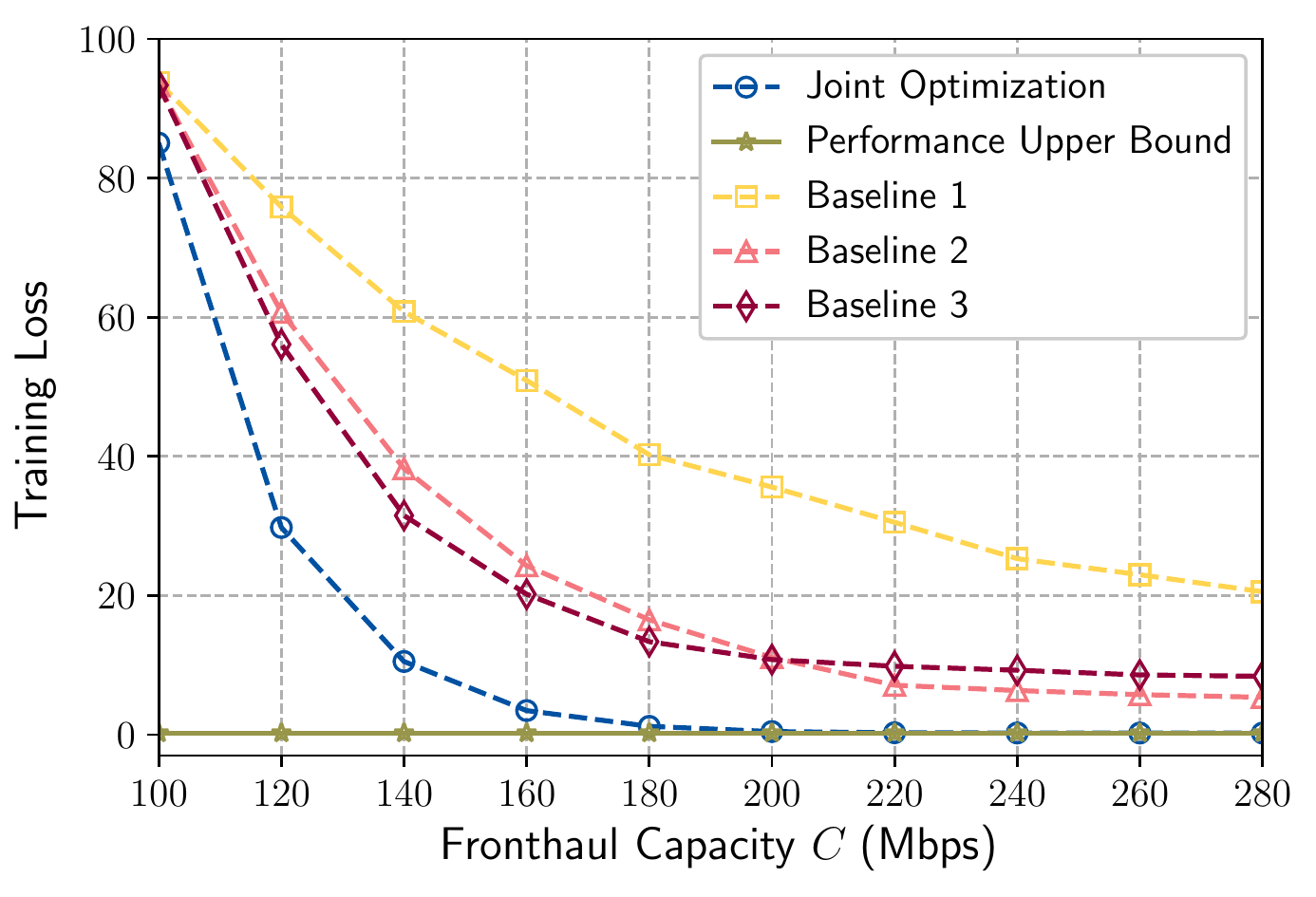}
	\end{subfigure}
	\begin{subfigure}[h]{0.49\textwidth}
	        \centering
	        \includegraphics[scale=0.6]{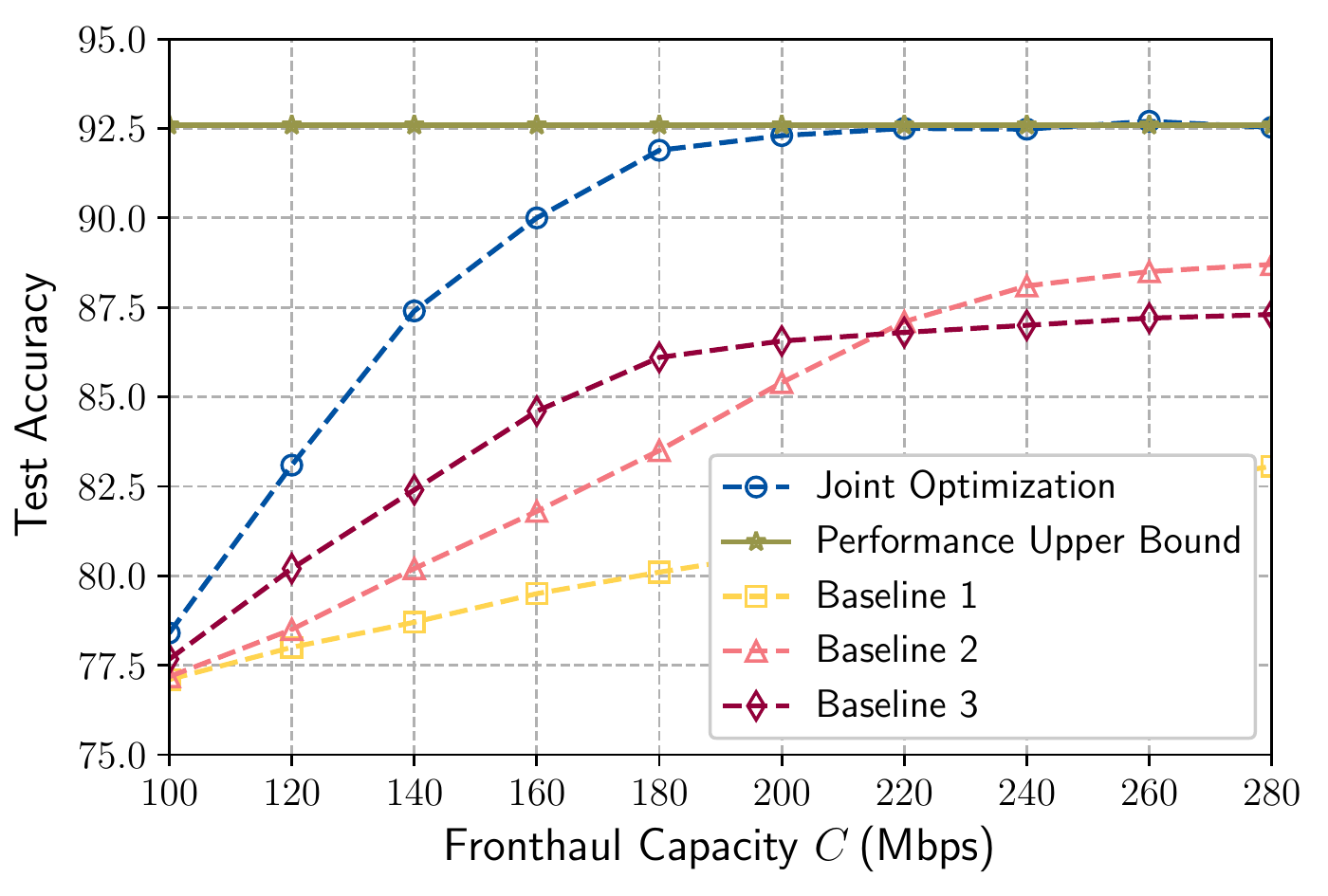}
	\end{subfigure}
	\vspace{-0.4cm}
	\caption{Training loss and test accuracy versus fronthaul capacity $C$ with $N=8$ and $M=2$.}
	\label{fig: capacity}
	\vspace{-0.8cm}
\end{figure}
	
Next, we investigate the impact of the antennas at each edge server on the performance of vertical FL, where the  fronthaul capacity is set to $C = 200$ Mbps.
Fig. \ref{fig: antennas} shows that the joint optimization for beamforming and quantization allocation still outperforms Baselines 1-3, and achieves the performance upper bound with a few antennas (i.e., $M = 10$).
In addition, we note that when $M$ is small (i.e., $M \leq 10$), Baseline 2 has better performance compared with Baseline 3, which is due to the performance gain of beamforming optimization dominates in this stage.
However, as the number of antennas $M$ increases, the performance gap between Baseline 2 and Baseline 3 vanishes.
The uniform quantization scheme used in Baseline 2 leads to severe quantization error since the capacity allocated to each antenna is very limited, which counteracts the performance gain of optimized beamforming.

\begin{figure}[h]
	\centering
	\begin{subfigure}[h]{0.49\textwidth}
	        \centering
	        \includegraphics[scale=0.6]{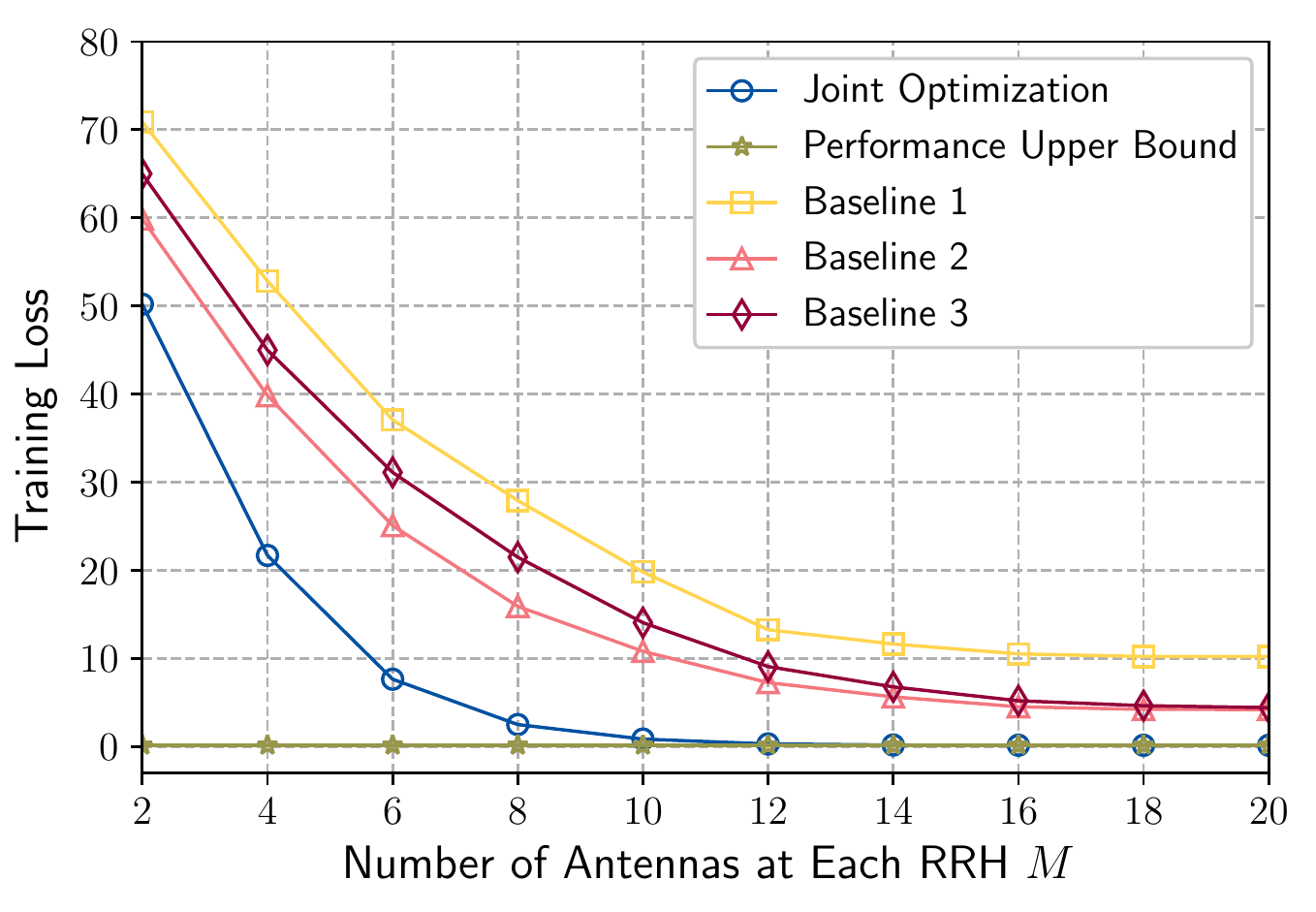}
	\end{subfigure}
	\begin{subfigure}[h]{0.49\textwidth}
	        \centering
	        \includegraphics[scale=0.6]{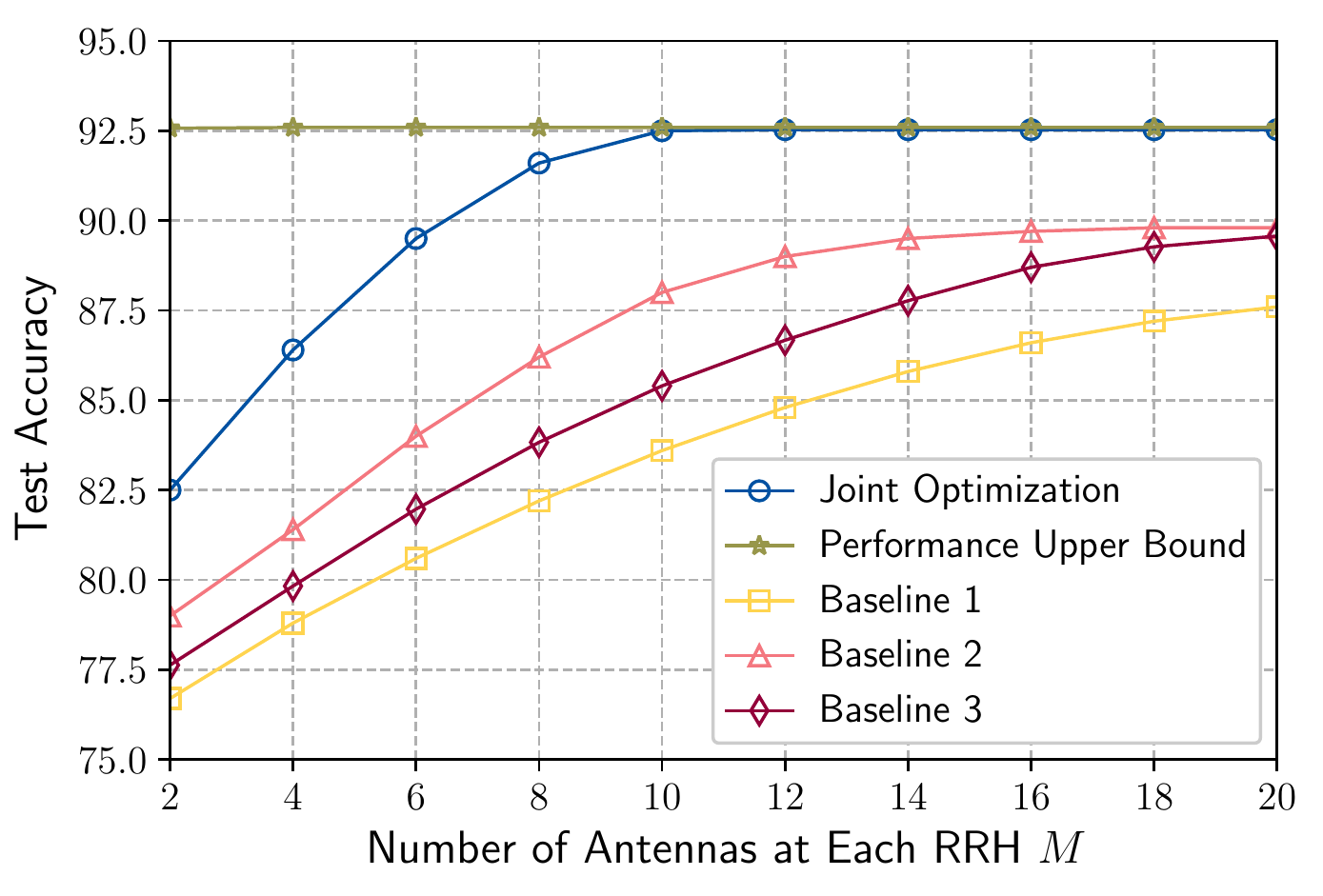}
	\end{subfigure}
	\vspace{-0.4cm}
	\caption{Training loss and test accuracy versus number of antennas of at each edge server $M$ with $N=8$ and $C=200$ Mbps.}
	\label{fig: antennas}
\end{figure}

\subsection{Cloud-RAN versus Massive MIMO}
In this subsection, we compare two promising techniques among the antenna configuration proposed for 5G wireless networks, i.e., Cloud-RAN and Massive MIMO.
The former is a distributed antenna system, while the latter is a centralized antenna system.
% Our goal is to find out which configuration we should equip given a total amount of antennas to be deployed.
For comparison fairness, we assume that there are a total of $\tilde{M}$ antennas to serve $K$ devices.
In particular, for the massive MIMO network, we assume that there is only one BS equipped with all $\tilde{M}$ antennas, while for the Cloud-RAN network, we assume that there are $N$ edge servers each equipped with $\tilde{M}/N$ antennas.
In order to compare the two configurations, we eliminate the limited fronthaul constraint in the Cloud-RAN network by setting $C = +\infty$, which leads to zero quantization errors when aggregating.
The specific simulation settings are as follows: 
$K = 20$ devices are randomly located in a circular area of $800$ m. 
There are a total number of antennas $\tilde{M} = 40$. 
For the massive MIMO network, the BS is located in the center of the circular area;
while for the Cloud-RAN network, the edge servers are randomly located in the circular area.
Furthermore, the beamforming vectors $\{\bm{m}^{(t)}, \bm{u}^{(t)}\}$ for the considered networks are obtained by Algorithm \ref{algorithm 2} and Algorithm \ref{algorithm 3} with $C = +\infty$.
In order to record the performance of each network, each numerical experiment runs $500$ communication rounds, and the performance of the last round is recorded.

Fig. \ref{fig: rrhs} shows the learning performance comparison between Cloud-RAN and massive MIMO with different number of edge servers.
It is observed that when there is only one edge server, the performance of the Cloud-RAN network is slightly worse than that of the massive MIMO network.
However, when the number of edge servers $N \geq 2$, the performance over the Cloud-RAN network outperforms that over the massive MIMO network, which is due to that the gain of distributed antenna system dominates the performance of Cloud-RAN network.
Specifically, by deploying more edge servers in the network, each device can be served by the nearest edge server with strong channel conditions, which significantly mitigates communication straggler issues.
A larger number of RRHs results in better channel condition due to dense deployment of RRHs.
However, the capacity allocated on each server will be small, resulting in a large quantization error on the learning performance, vice versa.
This leads to a trade-off between the channel conditions and quantization error in multi-antenna Cloud-RAN, which allows an optimization of the number of RRHs for maximizing the learning performance.
% This is also consistent with the conclusions reached in \cite{DBLP:journals/tsp/LiuZ15}.
As the total capacity increases, we should allocate more edge servers to exploit better channels for devices to improve performance.

\begin{figure}[h]
\vspace{-0.5cm}
\centering
\begin{subfigure}[h]{0.49\textwidth}
        \centering
        \includegraphics[scale=0.6]{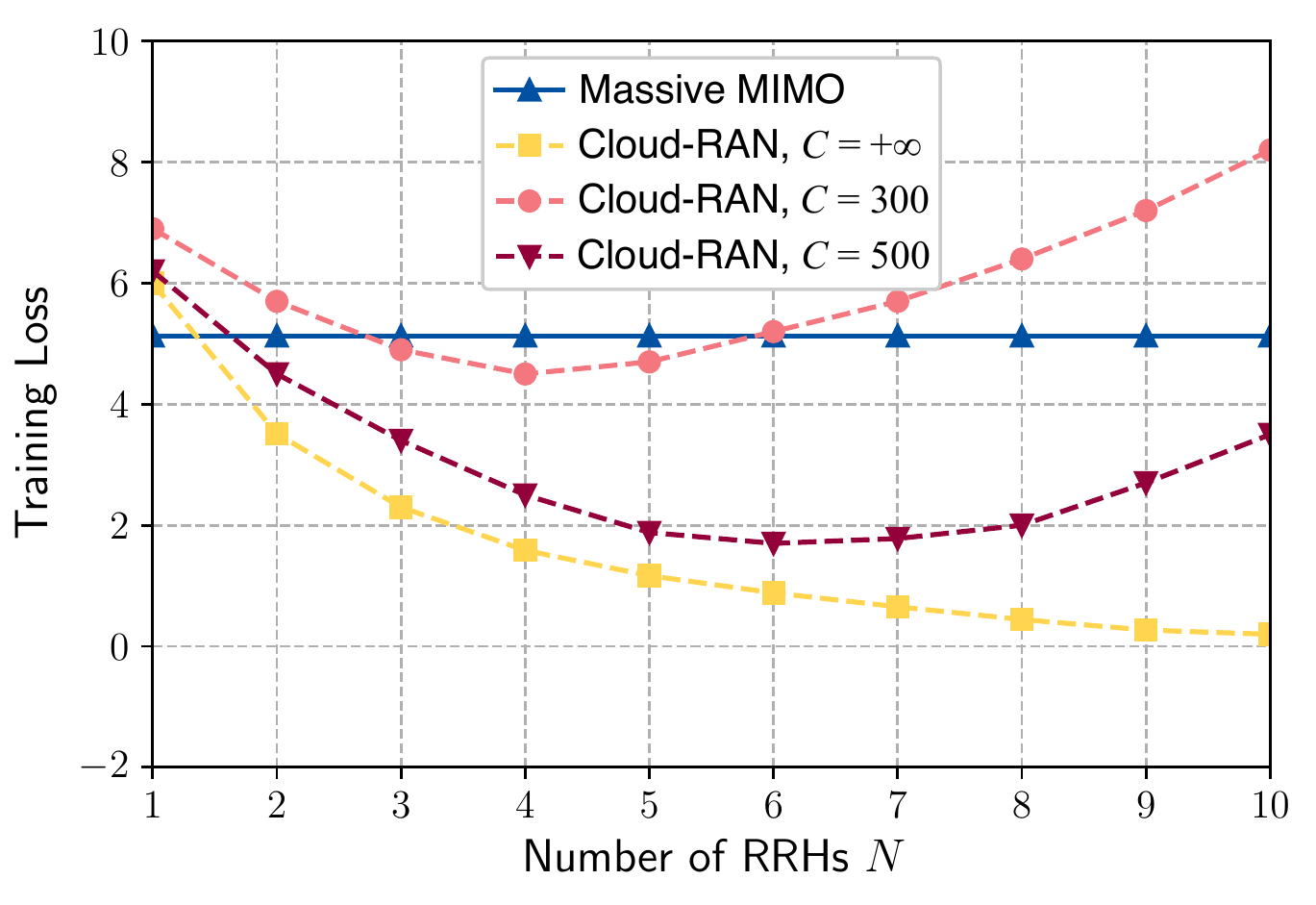}
\end{subfigure}
\begin{subfigure}[h]{0.49\textwidth}
        \centering
        \includegraphics[scale=0.6]{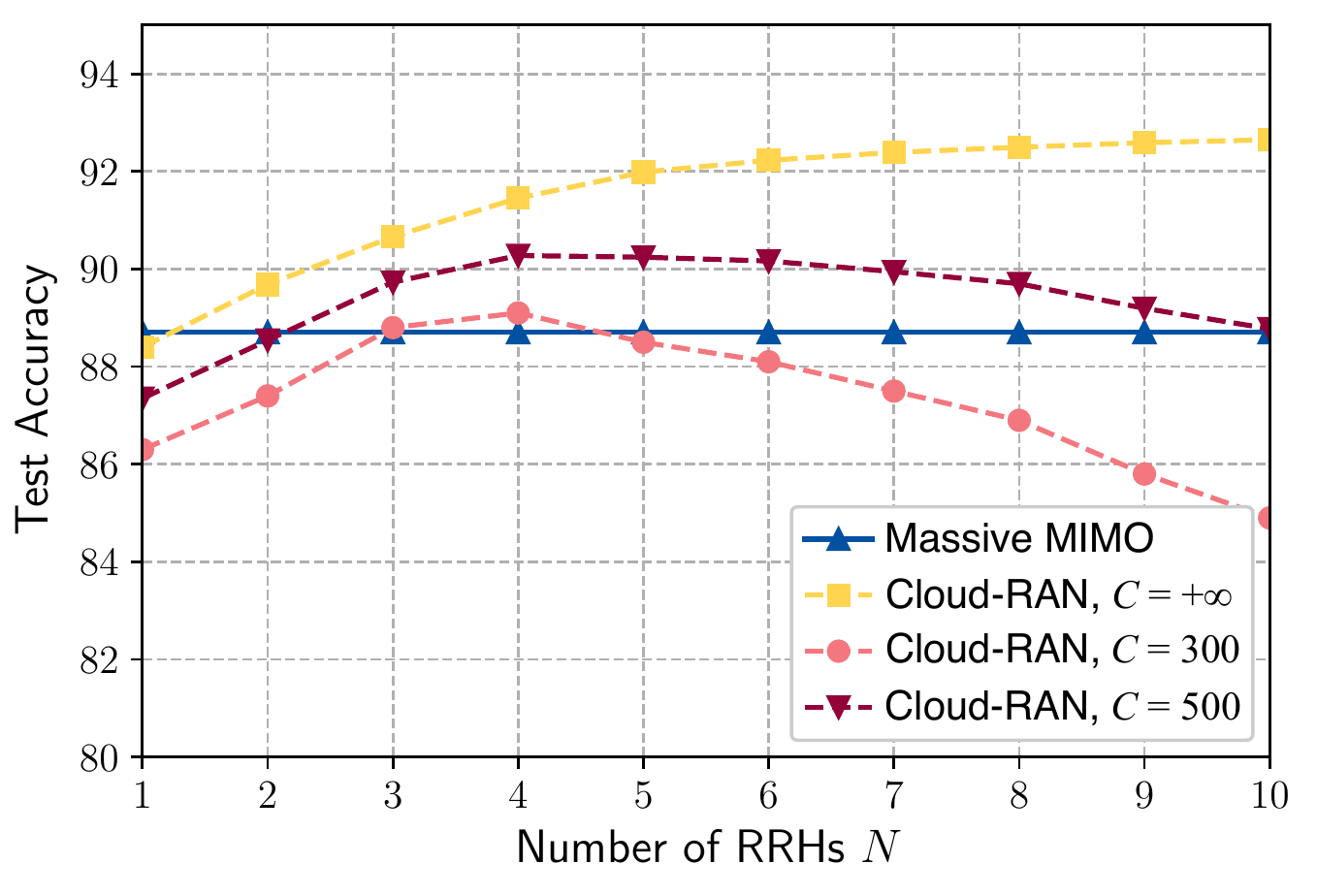}
\end{subfigure}
\vspace{-0.4cm}
\caption{Training loss and test accuracy versus number of RRHs $N$.}
\label{fig: rrhs}
\vspace{-1cm}
\end{figure}

\section{Conclusion}
In this paper, we proposed a novel framework based on AirComp and Cloud-RAN to support communication-efficient vertical FL over wireless networks. 
To reveal the impact of limited fronthaul capacity and AirComp aggregation error on the learning performance, we characterized the convergence behavior in terms of resource allocation. 
Based on the derived optimization gap for the wireless vertical FL algorithm, we further established a system optimization framework to minimize optimization gap for the wireless vertical FL algorithm under various system constraints. 
To solve this problem, we proposed to decompose the original system optimization problem into two sub-problems, i.e., uplink and downlink optimization problems, which can be efficiently solved by successive convex approximation and alternate convex search approaches. 
Extensive numerical results have shown that the proposed system optimization schemes can achieve high learning performance and the effectiveness of the proposed Cloud-RAN network architecture for vertical FL was also verified.

%\newpage
\appendices
\section{Proof of Theorem 1}\label{sec: proof of theorem 1}
According to \eqref{eq: partial gradient} and Lemma \ref{lemma}, the partial gradient estimate $\hat{\nabla}_k F(\bm{w}^{(t)})$ is given by
\begin{align}\label{eq: biased local gradient}
\hat{\nabla}_k F(\bm{w}^{(t)}) =& \frac{1}{L} \sum_{i=1}^L \hat{G}_i\left(\hat{s}^{(t)}(i)\right) \nabla g_{k,i}(\bm{w}_k)  + \lambda \nabla r_k(\bm{w})\nonumber \\
    =& \frac{1}{L} \sum_{i=1}^L \left(G_i\left(s^{(t)}(i)\right) + G'_i\left(s^{(t)}(i)\right)n^{(t)}_{\mathrm{UL}}(i) + n^{(t)}_{\mathrm{DL},k}(i)\right) \nabla g_{k,i}(\bm{w}_k) + \lambda \nabla r_k(\bm{w}) \nonumber \\
    % =& \nabla_k F(\bm{w}^{(t)}) + \frac{1}{L} \sum_{i=1}^L \left(G'_i\left(s^{(t)}(i)\right)n^{(t)}_{\mathrm{UL}}(i) + n^{(t)}_{\mathrm{DL},k}(i)\right) \nabla g_{k,i}(\bm{w}_k) \nonumber \\
    =& \nabla_k F(\bm{w}^{(t)}) + \bm{e}^{(t)}_k,
\end{align}
where $\bm{e}^{(t)}_k$ denotes the effective communication noise for device $k$ in the $t$-th communication round that is given by
\begin{align}
	\bm{e}^{(t)}_k = \frac{1}{L} \sum_{i=1}^L \left(G'_i\left(s^{(t)}(i)\right)n^{(t)}_{\mathrm{UL}}(i) + n^{(t)}_{\mathrm{DL},k}(i)\right) \nabla g_{k,i}(\bm{w}_k).
\end{align}
In addition, the expected norm of $\bm{e}^{(t)}$ over the noises $n^{(t)}_{\mathrm{UL}}(i)$ and $n^{(t)}_{\mathrm{DL},k}(i)$ can be bounded as follows
\begin{align}
	\mathbb{E}[\|\bm{e}^{(t)}_k\|^2] =& \frac{1}{L^2}\mathbb{E}\left[\left\|\sum_{i=1}^L\left(G'_i\left(s^{(t)}(i)\right)n^{(t)}_{\mathrm{UL}}(i) + n^{(t)}_{\mathrm{DL},k}(i)\right) \nabla g_{k,i}(\bm{w}_k)\right\|^2\right] \nonumber \\
	\overset{(\mathrm{i})}{=}& \frac{1}{L^2}\sum_{i=1}^L\mathbb{E}\left[\left(G'_i\left(s^{(t)}(i)\right)n^{(t)}_{\mathrm{UL}}(i) + n^{(t)}_{\mathrm{DL},k}(i)\right)^2\right]\left\|\nabla g_{k,i}(\bm{w}_k)\right\|^2 \nonumber \\
	% =& \frac{1}{L^2}\sum_{i=1}^L\left(\mathbb{E}\left[\left(G'_i\left(s^{(t)}(i)\right)n^{(t)}_{\mathrm{UL}}(i)\right)^2\right] + \mathbb{E}\left[(n^{(t)}_{\mathrm{DL},k}(i))^2\right]\right)\left\|\nabla g_{k,i}(\bm{w}_k)\right\|^2 \nonumber \\
	% =& \frac{1}{L^2}\sum_{i=1}^L\left(G'_i\left(s^{(t)}(i)\right)^2\sigma^2_{\mathrm{UL}}(t) + \sigma^2_{\mathrm{DL},k}(t)\right)\left\|\nabla g_{k,i}(\bm{w}_k)\right\|^2 \nonumber \\
	=& \frac{1}{L^2}(\Phi_{1,k}\sigma^2_{\mathrm{UL}}(t) + \Phi_{2,k}\sigma^2_{\mathrm{DL},k}(t)),
\end{align}
where
\begin{align}
	\Phi_{1,k} = \sum_{i=1}^LG'_i\left(s^{(t)}(i)\right)^2\left\|\nabla g_{k,i}(\bm{w}_k)\right\|^2, \quad 
	\Phi_{2,k} = \sum_{i=1}^L\left\|\nabla g_{k,i}(\bm{w}_k)\right\|^2.
\end{align}
Let $\bm{w}^{(t+1)} = [(\bm{w}^{(t+1)}_1)^\mathsf{T}, \ldots, (\bm{w}^{(t+1)}_k)^\mathsf{T}]^\mathsf{T}$, and use the definition of $\bm{w}^{(t+1)}_k$ to obtain
\begin{align}
\bm{w}^{(t+1)} = \left[\left(\bm{w}^{(t)}_1 - \mu^{(t)} \hat{\nabla}_1F(\bm{w}^{(t)})\right)^\mathsf{T}, \ldots, \left(\bm{w}^{(t)}_K - \mu^{(t)} \hat{\nabla}_k F(\bm{w}^{(t)})\right)^\mathsf{T}\right]^\mathsf{T} 
% =& [(\bm{w}^{(t)}_1)^\mathsf{T}, \ldots, (\bm{w}^{(t)}_K)^\mathsf{T}]^\mathsf{T} - \mu^{(t)} \left[\hat{\nabla}_1F(\bm{w}^{(t)}), \ldots, \hat{\nabla}_k F(\bm{w}^{(t)})\right] \nonumber \\
= \bm{w}^{(t)} - \mu^{(t)} \hat{\nabla} F(\bm{w}^{(t)}),
\end{align}
where 
\begin{align}
	\hat{\nabla} F(\bm{w}^{(t)}) = 
	% \left[\left(\hat{\nabla}_1F(\bm{w}^{(t)})\right)^\mathsf{T}, \ldots, \left(\hat{\nabla}_k F(\bm{w}^{(t)})\right)^\mathsf{T}\right]^\mathsf{T} \nonumber \\
	% =& \left[\left(\nabla_1 F(\bm{w}^{(t)}) + \bm{e}^{(t)}_1\right)^\mathsf{T}, \ldots, \left(\nabla_K F(\bm{w}^{(t)}) + \bm{e}^{(t)}_K\right)^\mathsf{T}\right]^\mathsf{T} \nonumber \\
	\underbrace{\left[\left(\nabla_1 F(\bm{w}^{(t)})\right)^\mathsf{T}, \ldots, \left(\nabla_K F(\bm{w}^{(t)})\right)^\mathsf{T}\right]^\mathsf{T}}_{=\nabla F(\bm{w}^{(t)})} + \underbrace{\left[\left(\bm{e}^{(t)}_1\right)^\mathsf{T}, \ldots, \left(\bm{e}^{(t)}_K\right)^\mathsf{T}\right]^\mathsf{T}}_{=\bm{e}^{(t)}}.
\end{align}
Here, the equality (i) comes from \eqref{eq: biased local gradient}.

Under Assumption \ref{ass:2}, we have
\begin{align}\label{eq: smooth bound}
F(\bm{w}^{(t+1)}) &\leq F(\bm{w}^{(t)}) + \nabla F(\bm{w}^{(t)})^{\mathsf{T}} (\bm{w}^{(t+1)} - \bm{w}^{(t)})  + \frac{\beta}{2} \| \bm{w}^{(t+1)} - \bm{w}^{(t)} \|^2 \nonumber \\
&= F(\bm{w}^{(t)}) - \mu^{(t)} \nabla F(\bm{w}^{(t)})^{\mathsf{T}} \hat{\nabla} F(\bm{w}^{(t)})  +\frac{(\mu^{(t)})^2 \beta}{2} \left\| \hat{\nabla} F(\bm{w}^{(t)}) \right\|^2 \nonumber \\
% &= F(\bm{w}^{(t)}) - \mu^{(t)} \left\| \nabla F(\bm{w}^{(t)}) \right\|^2 - \mu^{(t)} \nabla F(\bm{w}^{(t)})^{\mathsf{T}} \bm{e}^{(t)} \nonumber \\
% & \quad + \fr	ac{(\mu^{(t)})^2 \beta}{2} \left\| \nabla F(\bm{w}^{(t)}) \right\|^2 + (\mu^{(t)})^2 \beta\nabla F(\bm{w}^{(t)})^{\mathsf{T}} \bm{e}^{(t)}\notag + \frac{(\mu^{(t)})^2 \beta}{2} \| \bm{e}^{(t)} \|^2 \nonumber \\
&= F(\bm{w}^{(t)}) - \frac{1}{2\beta}  \left\| \nabla F(\bm{w}^{(t)}) \right\|^2 + \frac{1}{2\beta} \| \bm{e}^{(t)} \|^2,
\end{align}
where the learning rate $\mu^{(t)} = 1/\beta$ implies the last equality. 

In order to derive a lower bound for $\left\| \nabla F(\bm{w}^{(t)}) \right\|^2$, we use Assumption \ref{ass:1} and minimize both sides of \eqref{eq: convex} to obtain
\begin{align}
	% F(\bm{w}^*) \geq& F(\bm w)+\nabla F(\bm w)^\mathsf{T}(\bm w - \bm w^*) + \frac{\alpha}{2} \|\bm w - \bm w^*\|^2 \nonumber \\
	% \geq& \min_{\bm w} \left\{F(\bm w)+\nabla F(\bm w)^\mathsf{T}(\bm w - \bm w^*) + \frac{\alpha}{2} \|\bm w - \bm w^*\|^2\right\} \nonumber \\
	% =& F(\bm w) - \frac{1}{\alpha}\nabla F(\bm{w})^\mathsf{T} \nabla F(\bm{w}) + \frac{1}{2\alpha}\nabla F(\bm{w})^\mathsf{T} \nabla F(\bm{w}) \nonumber \\
	% =& F(\bm{w}) - \frac{1}{2\alpha}\left\|\nabla F(\bm{w})\right\|^2.
	F(\bm{w}^*) \geq F(\bm w) - \frac{1}{\alpha}\nabla F(\bm{w})^\mathsf{T} \nabla F(\bm{w}) + \frac{1}{2\alpha}\nabla F(\bm{w})^\mathsf{T} \nabla F(\bm{w}) = F(\bm{w}) - \frac{1}{2\alpha}\left\|\nabla F(\bm{w})\right\|^2.
\end{align}
Rearranging the terms, we have
\begin{align}\label{eq: lower bound of local gradient}
	\left\|\nabla F(\bm{w}^{(t)})\right\|^2 \geq 2\alpha\left(F(\bm{w}^{(t)}) - F(\bm{w}^*)\right).
\end{align}
Subtract $F(\bm{w}^*)$ from both sides and substitute \eqref{eq: lower bound of local gradient} into \eqref{eq: smooth bound} to get
\begin{align}
	F(\bm{w}^{(t+1)}) - F(\bm{w}^*)
	\leq&  F(\bm{w}^{(t)}) - F(\bm{w}^*) - \frac{\alpha}{\beta}  \left(F(\bm{w}^{(t)}) - F(\bm{w}^*)\right) + \frac{1}{2\beta} \| \bm{e}^{(t)} \|^2 \nonumber \\
	=& \left(1-\frac{\alpha}{\beta}\right)  \left(F(\bm{w}^{(t)}) - F(\bm{w}^*)\right) + \frac{1}{2\beta} \| \bm{e}^{(t)} \|^2
\end{align}
Applying this recursively, we obtain
\begin{align}\label{eq: original convergence}
	F(\bm{w}^{T}) - F(\bm{w}^*) 
	\leq& \rho^T\left(F(\bm{w}^0) - F(\bm{w}^*)\right) + \frac{1}{2\beta}\sum_{t=0}^{T-1}\rho^{T-t-1} \| \bm{e}^{(t)} \|^2,
\end{align}
where $\rho = 1 - \alpha / \beta$ is the contraction rate.

Taking expectations of \eqref{eq: original convergence} over the noises $n^{(t)}_{\mathrm{UL}}(i)$ and $\{n^{(t)}_{\mathrm{DL},k}(i)\}_{k=1}^K$, we obtain
\begin{align}
	\mathbb{E}\left[F(\bm{w}^{T}) - F(\bm{w}^*)\right]
	\leq& \rho^T\mathbb{E}\left[F(\bm{w}^0) - F(\bm{w}^*)\right]  + \frac{1}{2\beta}\sum_{t=0}^{T-1}\rho^{T-t-1} \mathbb{E}[\|\bm{e}^{(t)}\|^2] \nonumber \\
	\leq& \rho^T\mathbb{E}\left[F(\bm{w}^0) - F(\bm{w}^*)\right] + \frac{1}{2L^2\beta} B(T),
\end{align}
where the optimality gap $B(T)$ is given by
\begin{align}
	B(T) = \sum_{t=0}^{T-1}\rho^{T-t-1}\sum_{k=1}^K(\Phi_{1,k}\sigma^2_{\mathrm{UL}}(t) + \Phi_{2,k}\sigma^2_{\mathrm{DL},k}(t)),
\end{align}
which completes the proof of Theorem \ref{the1}.

% \section{Proof of Proposition \ref{prop1}} \label{proof: prop1}
% We first reformulate problem \eqref{prob: uplink optimization2} into the following equivalent problem:
% \begin{equation}
% \min_{ \bm{m}} \max_{k} \quad -\frac{{|\bm{m}^\mathsf{H}\bm{h}_{\mathrm{ul,k}}|^2}}{\bm{m}^\mathsf{H}\tilde{\bm{Q}}\bm{m}}.
% \end{equation}
% By introducing an auxiliary variable $\tau = \bm{m}^\mathsf{H}\tilde{\bm{Q}}\bm{m}$, we can rewrite the above problem as 
% \begin{equation}
% 	\begin{aligned}
% \min_{ \bm{m}} \max_{k} \quad -\frac{{|\bm{m}^\mathsf{H}\bm{h}_{\mathrm{ul,k}}|^2}}{\tau},\quad
% \mathrm{s.t.} \quad \bm{m}^\mathsf{H}\tilde{\bm{Q}}\bm{m} = \tau.
% \end{aligned}
% \end{equation}
% Defining $\tilde{\bm{m}} = \bm{m} / \sqrt{\tau}$, the problem becomes
% \begin{equation}
% \begin{aligned}
% \min_{ \tilde{\bm{m}}} \max_{k} \quad -|\bm{m}^\mathsf{H}\bm{h}_{\mathrm{ul,k}}|^2, \quad
% \mathrm{s.t.} \quad \tilde{\bm{m}}^\mathsf{H}\tilde{\bm{Q}}\tilde{\bm{m}} = 1,
% \end{aligned}
% \end{equation}
% which completes the proof of Proposition \ref{prop1}.

\footnotesize
\bibliographystyle{IEEEtran}
\bibliography{refs.bib}

\end{document}